\documentclass[10pt]{article}
\newcommand{\lint}{\llbracket}
\newcommand{\rint}{\rrbracket}
\usepackage{stmaryrd}
 \usepackage[usenames, dvipsnames]{xcolor}
\usepackage{tikz}
\usetikzlibrary{patterns}
\usepackage[utf8]{inputenc}
\usepackage{amsmath, latexsym, amsfonts, amssymb, amsthm}
\usepackage{graphicx,hyperref,dsfont,epsfig,caption,wrapfig,subfig}
\usepackage{datetime}
\usepackage{tikz}
\usepackage{movie15}
\hypersetup{
    linktoc=page,
    linkcolor=red,          
    citecolor=blue,        
    filecolor=blue,      
    urlcolor=cyan,
    colorlinks=true           
}

\newcommand{\cA}{{\ensuremath{\mathcal A}} }
\newcommand{\cB}{{\ensuremath{\mathcal B}} }

\newcommand{\cF}{{\ensuremath{\mathcal F}} }

\newcommand{\cL}{{\ensuremath{\mathcal L}} }
\newcommand{\cT}{{\ensuremath{\mathcal T}} }
\newcommand{\cD}{{\ensuremath{\mathcal D}} }

\newcommand{\cG}{{\ensuremath{\mathcal G}} }

\newcommand{\demi}{\tfrac{1}{2}}
\newcommand{\cjd}{\rangle}
\newcommand{\mc}{\mathcal}

\newcommand{\la}{\lambda}
\newcommand{\cjg}{\langle}

\newcommand{\bP}{{\ensuremath{\mathbf P}} }
\newcommand{\bE}{{\ensuremath{\mathbf E}} }

\newcommand{\bbC}{{\ensuremath{\mathbb C}} }

\newcommand{\bbE}{{\ensuremath{\mathbb E}} }

\newcommand{\bbP}{{\ensuremath{\mathbb P}} }

\newcommand{\bbR}{{\ensuremath{\mathbb R}} }

\newcommand{\gb}{\beta}
\newcommand{\gep}{\varepsilon}       
\newcommand{\gD}{\Delta}
\numberwithin{equation}{section}

\newtheorem{theorem}{Theorem}[section]
\newtheorem{defth}[theorem]{Definition-Theorem} 
\newtheorem{lemma}[theorem]{Lemma}
\newtheorem{proposition}[theorem]{Proposition}

\newtheorem{rem}[theorem]{Remark}
\newtheorem{remark}[theorem]{Remark}

\newtheorem{assumption}[theorem]{Assumption}

\theoremstyle{definition}

\renewcommand{\tilde}{\widetilde}          
\DeclareMathSymbol{\leqslant}{\mathalpha}{AMSa}{"36} 
\DeclareMathSymbol{\geqslant}{\mathalpha}{AMSa}{"3E} 
\DeclareMathSymbol{\eset}{\mathalpha}{AMSb}{"3F}     
\renewcommand{\leq}{\;\leqslant\;}                   
\renewcommand{\geq}{\;\geqslant\;}                   
\newcommand{\dd}{\text{\rm d}}             

\newcommand{\C}{\mathbb{C}}
\newcommand{\R}{\mathbb{R}}

\newcommand{\N}{\mathbb{N}}

\newcommand{\E}{\mathds{E}}
\renewcommand{\P}{\mathds{P}}
\newcommand{\ind}{\mathds{1}}

\renewcommand{\hat}{\widehat}

\usepackage{xparse}

\DeclareMathOperator*{\esssup}{ess\,sup}

\def\eps{\varepsilon}

\def\bi{\begin{itemize}}
\def\ei{\end{itemize}}

\def\bnum{\begin{enumerate}}
\def\enum{\end{enumerate}}
\def\<#1{\langle #1 \rangle}
\def\red#1{\textcolor{red}{#1}}

\def\cF{\mathcal{F}}

\definecolor{darkred}{rgb}{0.7,0.1,0.1}
\title{Path integral for quantum  Mabuchi K-energy}

\author{ Hubert Lacoin \footnote{IMPA, Estrada Dona Castorina 110, Rio de Janeiro, RJ-22460-320, Brasil. H.L. acknowledges support from FAPERj (grant JCNE) and CPNq (grant Universal and productivity grant).}, R\'emi Rhodes \footnote{Universit{\'e} Paris-Est Marne la Vall\'ee, LAMA, Champs sur Marne, France. Partially supported by grant    ANR-15-CE40-0013 Liouville.} \footnotetext[2]{Partially supported by grant    ANR-15-CE40-0013 Liouville.},
 Vincent Vargas \footnote{ENS Ulm, DMA, 45 rue d'Ulm,  75005 Paris, France. Partially supported by grant    ANR-15-CE40-0013 Liouville.} }

\begin{document}

\maketitle

\begin{abstract}
We construct a path integral based on the coupling of  the Liouville action and the Mabuchi K-energy on a one-dimensional complex manifold. To the best of our knowledge this is the first rigorous construction of such an object and this is done by means of probabilistic tools. 
Both functionals play an important role respectively in Riemannian geometry (in the case of surfaces) and  K\"ahler geometry. 
As an output, we obtain a path integral whose Weyl anomaly displays the standard Liouville anomaly plus an additional   K-energy term. Motivations come from theoretical physics where these type of path integrals arise as a model for fluctuating metrics on surfaces   when coupling (small) massive  perturbations of conformal field theories to quantum gravity as advocated by A. Bilal, F. Ferrari, S. Klevtsov and S. Zelditch.
Interestingly, our computations show that quantum corrections perturb  the classical  Mabuchi K-energy and produce a quantum Mabuchi K-energy: this type of correction is reminiscent of the quantum Liouville theory. Our construction is probabilistic and relies on a variant of Gaussian multiplicative chaos (GMC), the Derivative GMC (DGMC for short). The technical backbone of our construction consists in two estimates on (derivative and standard) GMC which are of independent interest in probability theory.  Firstly, we show that these DGMC random variables possess negative exponential moments  and secondly we derive optimal small deviations estimates for the GMC associated with a recentered Gaussian Free Field. 
\end{abstract}

\begin{center}
\end{center}
\footnotesize



\noindent{\bf Key words or phrases:}  $2d$ Quantum Gravity, quantum field theory, Gaussian multiplicative chaos, random K\"ahler geometry, quantum Mabuchi.

\noindent{\bf MSC 2000 subject classifications:  81T40,  81T20, 60D05.}    

\normalsize

%
\tableofcontents

\section{Introduction and motivations}
 
 The goal of this paper is to construct a new form of $2d$ random geometry (a so-called \textbf{quantum K\"ahler geometry}) using probabilistic methods in view of applications to $2d$ quantum gravity. Given a manifold $M$, quantum gravity is a geometrical prescription to pick  at random a geometry (say a metric tensor)   together with a matter field on $M$ (a configuration of some model of statistical physics). This prescription is a non trivial coupling of these two  objects so that the nature of the matter field shapes the random geometry. Yet this mechanism is   understood (at the level of physics rigor)   only in a very specific situation when the matter field possesses conformal symmetries, namely is a Conformal Field Theory (CFT for short)\footnote{Recall that many CFTs are expected to describe the scaling limit of discrete statistical physics model at criticality.}. In the physics literature, Polyakov's seminal work \cite{Pol} and the DDK ansatz \cite{cf:Da,DistKa} have paved the way towards a complete understanding of this case: the  random geometry is then ruled by the Liouville CFT (LCFT for short), which can be seen as the natural probabilistic theory of Riemannian geometry. Giving a rigorous meaning to this picture is under active research in probability theory nowadays and we will not try in this introduction to give an account on all the recent developments on the topic: we refer to \cite{DKRV,GRV} for rigorous constructions of LCFT and the works \cite{curien, DMS} for the link between LCFT and the scaling limit of discrete planar maps weighted by a critical statistical physics model. In the specific case of pure gravity where the matter field is trivial, one can also equip LCFT with a distance function called the Brownian map: see \cite{LeGall,Mier} for the convergence (in the sense of Gromov-Hausdorff) of discrete maps to the Brownian map and \cite{MS0, MS1,MS} for a construction of the Brownian map in the continuum setting of LCFT.  Our paper is concerned with the study of random geometries that may arise when matter fields (slightly) move away from conformal symmetries. Our approach is inspired by the series of works \cite{BFK,FKZ1,FKZ2} suggesting models deeply anchored in both Riemannian and K\"ahler geometries.

In the framework of Riemannian geometry we are given   a compact Riemann surface $M$, i.e. a one dimensional complex manifold\footnote{We will restrict in this paper to the case of one-dimensional complex manifolds. Though Riemann and K\"ahler geometries make sense in higher dimensions, our approach does not.}, equipped 
 with a Riemannian metric $g$. In view of classification of Riemann surfaces\footnote{This is also related to  solutions of Einstein field equations in general relativity.}, an important question  which goes back to Picard and Poincar\'e  is to find a metric $\hat g$ with uniformized Ricci scalar curvature $K_{\hat g}=-2\pi \mu$ in the conformal class $[g]$ of $g$ 
 $$[g]:=\{e^\omega g; \omega \in C^\infty (M)\}.$$
Such metrics $\hat g=e^\omega g$ can be found by searching for the critical points\footnote{In some cases, these critical points are also minimizers.} of the {\it classical} Liouville functional
\begin{equation}\label{intro:liouville}
S^{{\rm cl},\mu}_{{\rm L}}(\hat g,g):= \int_M\big(|d\omega|^2_{g}+ 2K_{g}\omega+4\pi\mu e^{\omega}\big){\rm dv}_{g},
\end{equation}
where $K_g$ is the Ricci scalar curvature of the metric $g$\footnote{In isothermal coordinates of the form $e^{\omega (z)} |dz|^2$, the curvature $K_g(z)$ is given by $- e^{- \omega(z)} \Delta_z \omega(z)$ where $\Delta_z$ is the standard flat Laplacian.}, ${\rm v}_{g}$ its volume form  and $d\omega$ the differential of $\omega$. More generally, in arbitrary dimensions,  metrics with  Ricci tensor proportional to the metric are called Einstein metrics. The problem of finding Einstein metrics on $2d$ real (or one dimensional complex) manifolds is now well understood but  turns out to be  much harder in higher dimensions. This has certainly been a source of motivation to put some further structure on the manifold   in order to make the search for Einstein metrics more tractable. This explains at least partly the success of K\"ahler geometry. Indeed, though K\"ahler geometry is a natural extension of Riemannian geometry in the sense that it is designed in the spirit of complex Euclidean geometry, it can also be seen as another parametrization of the set of metrics  that allows one to reduce the problem of finding Einstein metrics to a complex Monge-Amp\`ere equation, as illustrated by the works of Aubin \cite{aubin} or the proof of the Calabi-Yau conjecture \cite{Yau}. These works treat the cases of negatively curved or Ricci-flat manifolds. The case of positively curved manifolds, known as the Yau-Tian-Donaldson conjecture, has been more problematic and has been settled only recently by Chen-Donaldson-Sun in a series of works \cite{CDS} in which the Mabuchi K-energy presented below has played an important role: for further details we refer to the review paper \cite{Sze}.   From now on, we come back to the simpler framework of one dimensional complex manifolds and filter as much as possible geometrical considerations.

In the K\"ahler framework, the K\"ahler potential $\phi$ of the metric $\hat g=e^\omega g$ is defined by the relation
\begin{equation}\label{intro:kahler}
e^{\omega}=\frac{V_{\hat g}}{V_{g}}+\frac{V_{\hat g}}{2}\Delta_{g}\phi
\end{equation}
where $\Delta_{g}$ is the (negative) Laplace-Beltrami operator, with expression in local real coordinates $(x_1,x_2)$
$$\Delta_{g}=-\frac{1}{\sqrt{g}}\sum_{i,j=1}^2\frac{\partial}{\partial x_i}\Big(\sqrt{g} g^{ij}\frac{\partial}{\partial x_j}\,\Big),$$
 and $V_g:={\rm v}_{g}(M)$ is the volume of $M$ in the metric $g$. This equation can always be solved up to constant in $\phi$. An important functional called the Mabuchi K-energy can be written in terms of $\omega$ and  $\phi$
\begin{equation}\label{intro:mabuchi}
S^{{\rm cl}}_{\rm M}(\hat g,g)=\int_M\Big(2\pi(1-\textbf{h})\phi \Delta_g\phi+(\frac{8\pi(1-\textbf{h})}{V_{g}}-K_{g})\phi+\frac{2}{V_{\hat g}}\omega e^{\omega}\Big) d{\rm v}_{g}
\end{equation}
where $\mathbf{h}$ is the genus of $M$.
 Extremal points of the Mabuchi K-energy are also metrics with uniformized scalar curvature, hence the connection with K\"ahler-Einstein metrics.
 
The concept of uniformization of Riemann surfaces has its probabilistic  counterpart. Indeed, Feynman's approach of quantum mechanics prescribes to associate to the Liouville functional on the Riemannian manifold $(M,g)$  a path integral (i.e. a measure on some functional space) 
\begin{equation}\label{Liouvmeasintro}
\langle F\rangle_{{\rm L},g} =\int F(\varphi)e^{-\mathcal{S}_{\rm L}(\varphi,g)}\mathcal{D}\varphi
\end{equation}
 for suitable functionals $F$ , where $ \mathcal{D}\varphi$ is the putative uniform measure\footnote{This measure is called the free field measure in the physics literature (not to be confused with the Gaussian free field) but it not defined mathematically.} on some functional space of maps $\varphi:M\to\R$\footnote{For  (log)-conformal factors, we use throughout the papers two different notations: $\omega$ when it is deterministic and $\varphi$ when it serves as an integration variable.}    and $\mathcal{S}_{\rm L}$ is the {\it quantum} Liouville functional (in what follows, for practical purpose the quantum actions  are written  as functions of the conformal factor) 
\begin{equation}\label{QLiouville:intro}
\mathcal{S}_{\rm L}( \varphi,g):=\frac{1}{4\pi}\int_M \big(|d\varphi|^2_g+QK_g\varphi +4\pi \mu e^{\gamma\varphi}\big)\, {\rm dv}_g
\end{equation}
   where $\gamma$ is a positive parameter belonging to  $(0,2)$, $Q=\frac{\gamma}{2}+\frac{2}{\gamma}$ and $\mu>0$ is a positive parameter called  the cosmological constant (see subsection \ref{LCFT} for further details). This path integral   turns out to be a CFT, hence is called Liouville CFT. Such a path integral has  been   constructed {\bf non perturbatively} only very recently using probability theory (see \cite{DKRV,GRV}). This is in sharp contrast with many approaches to quantum field theory which usually provide constructions that are perturbative, i.e. are defined by formal power series (in the case of Liouville CFT on compact Riemann surfaces, the work of Takhtajan-Teo \cite{TT} provides a construction in terms of a formal power series in the parameter $\gamma$). Notice that the {\it quantum action} \eqref{QLiouville:intro} differs from the {\it classical action}  \eqref{intro:liouville} evaluated at $(\hat g,g)$  where $\hat g=e^{\gamma \omega}g$ 
\begin{equation}\label{intro:liouvillegamma}
\tfrac{1}{4\pi\gamma^2}S^{{\rm cl},\mu \gamma^2}_{{\rm L}}( \hat g,g):=\frac{1}{4\pi} \int_M\big(|d\omega|^2_{g}+ \frac{2}{\gamma}K_{g}\omega+4\pi\mu e^{\gamma\omega}\big){\rm dv}_{g}
\end{equation}
through the value of $Q$ where an extra $\frac{\gamma}{2}$ term appears:  this is due to quantum corrections appearing in renormalizing the theory (i.e. in controlling diverging quantities). 

The  Weyl anomaly describes the way a QFT reacts to conformal changes of metrics. In the case of Liouville CFT, it can be expressed in terms of the classical Liouville action: Consider a conformal metric $ \hat g=e^{\omega}g$ then     
\begin{equation}\label{weylCFTintro}
\langle F\rangle_{{\rm L},\hat g}= \langle F(\cdot\,-\tfrac{Q}{2}\omega)\rangle_{{\rm L}, g}  \exp\big(\frac{\mathbf{c}_{\rm L}}{96\pi}S^{{\rm cl},0}_{{\rm L}}( \hat g,g)\big) 
\end{equation}   
where $S^{{\rm cl},0}_{{\rm L}}$ is the classical Liouville functional (with $\mu=0$)
\begin{equation}\label{intro:liouville2}
S^{{\rm cl},0}_{{\rm L}}(\hat g,g):= \int_M\big(|d\omega|^2_{g}+ 2K_{g}\omega \big){\rm dv}_{g},
\end{equation}
and  $\mathbf{c}_{\rm L}=1+6Q^2$ is the {\it central charge} of Liouville CFT. The fact that the Weyl anomaly is log-proportional to the Liouville action  characterizes a CFT in general (up to regularity issues). Such a transformation rule encodes a great deal of information about the theory: In particular, conformal Ward identities come out of \eqref{weylCFTintro} (see \cite{gaw} for an argument for CFTs up to regularity issues and \cite{KRV} for a proof in the case of Liouville CFT), which leads to exact formulae for the theory (see in particular the recent proof of the DOZZ formula in \cite{KRV1}).

It is natural to wonder whether classical K\"ahler geometry admits a  probabilistic counterpart too. The purpose of this paper is to construct a path integral exhibiting a Mabuchi K-energy term in the Weyl anomaly based on the quantization of the Mabuchi K-energy. Motivations come from the need of understanding $2d$ quantum gravity coupled to non conformal QFT, which is translated in terms of random planar maps in appendix \ref{maps}. On the Riemannian manifold $(M,g)$,  this corresponds naively to constructing a functional integration measure of the type (with $\beta>0$ a coupling constant and $F$ an arbitrary functional)
\begin{equation}\label{wrong}
\int F(\varphi)e^{-\beta S^{{\rm cl}}_{\rm M}(e^{\gamma \varphi}g,g)- \mathcal{S}_{\rm L}(\varphi,g)}\mathcal{D}\varphi.
\end{equation}
It turns out that the proposal \eqref{wrong} does not possess the expected Weyl anomaly because it overlooks renormalization effects. Quantum corrections in the Mabuchi K-energy, reminiscent to those arising in the Liouville functional, force to consider instead the \textbf{quantum} Mabuchi action
\begin{equation}\label{introQmabuchi}
 \mathcal{S}_{\rm M}( \varphi,g)=\int_M\Big(2\pi(1-\textbf{h})\phi \Delta_g\phi+(\frac{8\pi(1-\textbf{h})}{V_{g}}-K_{g})\phi+\frac{2}{1-\frac{\gamma^2}{4}}\frac{1}{V_{\hat g}}(\gamma \varphi) e^{\gamma \varphi}\Big) {\rm dv}_{g}. 
\end{equation}
where  $\phi $ is the K\"ahler potential of the metric $e^{\gamma\varphi}$ (see Section \ref{sec:mabuchi} for precise definitions). Though the quantum versions of the Liouville and Mabuchi actions depend on $\gamma$, we most of the time do not stress the dependence in the notation. Compared to the classical Mabuchi K-energy \eqref{intro:mabuchi}, one can notice a quantum correction term $1-\frac{\gamma^2}{4}$ in the entropic term. As far as we know, this is the first  occurrence  of the quantum version of the K-energy in the literature. Now, the main input of the paper is to  construct the path integral 
\begin{equation}\label{MLintro}
\langle F\rangle_{{\rm ML},g} =\int F(\varphi)e^{-\beta  \mathcal{S}_{\rm M}( \varphi,g)- \mathcal{S}_{\rm L}(\varphi,g)}\mathcal{D}\varphi.
\end{equation}
 Compared to the Liouville path integral which corresponds to $\beta=0$, there is a serious extra difficulty in defining \eqref{MLintro} due to the potential term $(\gamma \varphi) e^{\gamma \varphi}$ in \eqref{introQmabuchi}. Making sense of  \eqref{MLintro} requires controlling this term from below, a non trivial task due to renormalization effects. Also the path integral \eqref{MLintro} has the expected Weyl anomaly (see Theorem \ref{main} for a precise statement): Let $\hat g=e^{\omega}g$ be a metric conformal to $g$ and denote by $\phi$ its K\"ahler potential   w.r.t. $g$.  Then
\begin{equation}
\langle F\rangle_{{\rm ML},  \hat g}= \langle F(\cdot-\tfrac{Q}{2}\omega)\rangle_{{\rm ML},g}\times \exp\big( \tfrac{1+6Q^2}{96\pi}S^{{\rm cl},0}_{{\rm L}}(\hat g,g)+\beta S^{{\rm cl}}_{\rm M}(\hat g, g)\big)
\end{equation}
where $S^{{\rm cl},0}_{{\rm L}}$ and $S^{{\rm cl}}_{\rm M}$ are respectively  the classical Liouville functional \eqref{intro:liouville2} and the classical Mabuchi K-energy \eqref{intro:mabuchi}. This path integral  is a way of giving sense to a random geometry of K\"ahler type.  In particular, the volume of the manifold is then promoted to a random variable: we prove that it has a Gamma law $\Gamma (s,\mu)$ with an explicit formula for $s$ in terms of $\gamma,\beta$ (see subsection \ref{sub:string}). This parameter $s$ thus appears as an area scaling exponent: it plays an important role in physics where it is called \textit{string susceptibility}\footnote{In fact, the string susceptibility is equal to $s+2$.}. Physicists do not have necessarily access to exact expressions for the string susceptibility so that they usually perform a \textit{loop expansion}, which is simply an asymptotic expansion of $s$ as $\gamma\to 0$. Our  exact formula for the string susceptibility reproduces exactly the loop expansion found in Bilal-Ferrari-Klevtsov \cite{BFK}, see subsection \ref{sub:string}. This is  somewhat striking as  our formula for $s$ is shaped by the quantum corrections in the Mabuchi action, whereas the computations  in \cite{BFK} do not rely on the same path integral approach.

\medskip
Our construction  is based on  Gaussian Multiplicative Chaos (GMC for short) as well as a variant. GMC Theory enables to define the exponential of the Gaussian Free Field (GFF).    We have chosen, for the sake of presentation, to introduce the rigorous and technical definitions behind our construction of the path integral \eqref{MLintro} only in Section \ref{sec:mabuchi}, but let us just mention that the construction is based on interpreting $e^{-\frac{1}{4\pi}\int_M |d\varphi|^2_g {\rm dv}_g} \mathcal{D} \varphi$ as a GFF measure and expressing the other terms in the actions as functions of the GFF. With this in mind, the term $e^{\gamma \varphi}$ in the Liouville action \eqref{QLiouville:intro} gives rise to GMC and the $(\gamma \varphi)e^{\gamma \varphi}$ term in the Mabuchi action \eqref{introQmabuchi} gives rise to a derivative (with respect to $\gamma$) of GMC.

More precisely, consider a GFF (see section \ref{sec:GFF}) $X$ with zero average when integrating with respect to the volume-form associated with $g$. A GMC measure is a random Radon (positive) measure   $M_\gamma$ of the form
\begin{equation}\label{GMCintro}
M_\gamma(dx):= e^{\gamma X(x)-\frac{\gamma^2}{2}\E[X^2(x)]}\,{\rm v}_g(\dd x).
\end{equation}
This expression is only formal as the GFF is a random distribution (in the sense of Schwartz), hence it is not a fairly defined function; this can be seen at the level of the variance which satisfies $\E[X^2(x)]= \infty$.
Renormalizing this  into a meaningful expression is what GMC theory is aiming for and it was mainly developed by Kahane in the eighties \cite{cf:Kah} (or see also \cite{review}). In our context, it asserts that the quantity \eqref{GMCintro} is well defined and non trivial for $\gamma \in (0,2)$. Now consider what we call {\bf derivative GMC} \footnote{The name comes from the fact that \eqref{GMCintroder} can be obtained from \eqref{GMCintro} by differentiating with respect to the parameter $\gamma$.} (DGMC for short)
\begin{equation}\label{GMCintroder}
M'_\gamma(dx):= (X(x)-\gamma \E[X^2(x)])e^{\gamma X(x)-\frac{\gamma^2}{2}\E[X^2(x)]}\,{\rm v}_g(\dd x)
\end{equation}
 in order to make sense of the $(\gamma \varphi)e^{\gamma \varphi}$ term in the Mabuchi action \eqref{introQmabuchi}.

On the technical side, our approach involves  three  ingredients (the last two are  {\bf new}) related to GMC or DGMC: 
\begin{itemize}
\item[(1)] Universality of the meaning of \eqref{GMCintroder} with respect to cut-off regularizations of the GFF  $X$. Universality of GMC measures is well established and key ingredients for that are Kahane's convexity inequalities (see \cite{review}) as well as positivity of GMC measures. For DGMC we lose both of these properties; as a matter of fact, DGMC is not even a signed measure almost surely (except for the limiting case $\gamma=2$). Yet universality can be restored for $\gamma \in (0,\sqrt{2})$ by using $L^2$-computations. \item[(2)] Concentration methods to bound the left tail of  DGMC: given a ball $B$ we show
$$\forall v\geq 0,\quad \P(M'_\gamma(B)\leq -v)\leq 2e^{-cv^2}$$
for some constant $c>0$. The technical estimates we use restrict our statement   to the values $\gamma \in (0,1)$. Yet we stress that these restrictions  can most likely be removed with some consequent amount of technicalities (see section \ref{sec:D} for a precise statement).
\item[(3)]  An {\bf optimal} small deviation result for GMC: let us  recenter the GFF so it has zero spatial average with respect to the measure ${\rm v}_g$ over a set $S$, namely  $\tilde X=X-\tfrac{1}{{\rm v}_g(S)}\int_S X \dd {\rm v}_g$ and denote by $\tilde{M}_\gamma$  the GMC measure for the field $\tilde X$. Then for $\gamma\in (0,2)$ and $s\geq 0$ 
$$\P(\tilde{M}_\gamma(S)\leq v )\leq c\exp\big(-cv^{-\frac{4}{\gamma^2}}|\ln v|^\kappa\big)$$
for some $\kappa,c>0$. (see Section \ref{sec:small} for a precise statement). Small deviations for GMC have received a lot of attention recently \cite{cf:DuSh,nikulae,GHSS} and are crucial estimates in many contexts. In all these works, the tail corresponds (at best) to lognormal random variables because  the leading fluctuation term corresponds to that of the spatial average of the field.  With the recentering procedure described above, we explain this mechanism  and identify the lower order contribution. This result  should be  sharp (when ignoring the log-correction) as illustrated by exact density results obtained in \cite{remy,rezhu}.
\end{itemize}

\medskip We stress here that the technical restrictions in items (1) and (2) above prevent our statements   from covering the whole range of expected valid parameters $\gamma\in (0,2)$. Note in particular that the degeneracy of the quantum Mabuchi K-energy \eqref{introQmabuchi} for $\gamma=2$ (though Liouville CFT is well defined) is quite intriguing and generalizing our theory to the limiting case $\gamma=2$ perhaps involves introducing a $2$-nd order derivative GMC. Furthermore another restriction in our statements has a geometrical flavor: we only consider the case of hyperbolic surfaces, in which case $\mathbf{h}\geq 2$. This entails two simplifications: first  we avoid this way having to introduce conical singularities in the surface $M$ (recall for instance this is the case for Liouville CFT on the Riemann sphere \cite{DKRV}) and, second, the sign in front of the term $\phi \Delta_g \phi$ in the quantum Mabuchi K-energy \eqref{mabuchi} goes the easy way. It is not hard to see that this term does not rise issues for small values of $\gamma$ on surfaces with genus $0$ or $1$ but a full treatment can be more problematic. To keep the paper reasonably short, we restrict to hyperbolic surfaces but investigating the case of the Riemann sphere or tori seems definitely interesting and challenging.

\subsection*{Organization of the paper} 
In Section \ref{sec:backgr}, we introduce the necessary technical background required to construct the Liouville path integral  including uniformization of Riemannian surfaces,
Green Functions, Gaussian Free Field and Gaussian Multiplicative Chaos. 

In Section \ref{sec:mabuchi}, we present our construction  of the path integral with action given by the combination Mabuchi K-energy and Liouville Action  ( Definition-Theorem \ref{main}) and introduce the main technical result on which the construction relies (Theorem \ref{exp}).

Section \ref{sec:D/M} announces the technical backbone the proof of  Theorem \ref{exp}.   More precisely it reduces the statement to left tail estimates for derivative GMC, and optimal small deviations for GMC measures proven respectively in Section  \ref{sec:D} and \ref{sec:small}.

\subsection*{Acknowledgements} The authors wish to thank A. Bilal, S. Boucksom, H. Erbin, S. Klevtsov, S. Zelditch for enlightening discussions, which have led to the final version of this manuscript.
This work was initiated during a stay of R.R. and V.V. at IMPA, they acknowledge kind hospitality and support.

\section{Background and notations}\label{sec:backgr}

In this section, we list our notations and recall the precise definition of  Liouville CFT (LCFT) as given in  \cite{GRV}.  
\subsection{Convention and notations.} 
 Given a  Riemannian manifold $(M,g)$ we denote by ${\rm dv}_g$ the associated Riemannian measure, $V_g={\rm v}_g(M)$ the total volume, $\Delta_g$ the Laplace-Beltrami operator (negative definite Laplacian), $K_g$ the  scalar curvature, $B_g(x,r)$ the ball centered at $x$ with radius $r$ in the metric $g$. 

We use standard notations for the spaces $C^\infty(M)$ of smooth (i.e. infinitely differentiable) functions on $M$ and  $L^p(M)$ for the (equivalence classes) of $p$-th power integrable functions.  

The notation $\dd x$ stands for the Lebesgue measure and we write $\ln_+(x)$ for $\max(\ln x,0)$.

\subsection*{Hyperbolic surfaces}   
 Let $M$ be a connected compact surface of genus ${\bf h}\geq 2$ (without boundary). The set of smooth metrics on $M$ is a Fr\'echet manifold denoted by ${\rm Met}(M)$. Let $g\in {\rm Met}(M)$. The Gauss-Bonnet  formula asserts that 
\begin{equation}\label{GB} 
\int_{M}K_g{\rm dv}_g=4\pi\chi(M)
\end{equation}
where $\chi(M)=(2-2{\bf h})$ is the Euler characteristic.

The Fr\'echet space $C^\infty(M)$ acts on ${\rm Met}(M)$ by conformal multiplication 
$(\omega, g)\mapsto e^{\omega}g$. The orbits of this action are called \emph{conformal classes} and the conformal class of a metric  $g$ is denoted by $[g]$. For a metric $\hat{g}=e^{\omega}g$, one has the relation 
\begin{equation}\label{curvature}
K_{\hat{g}}=e^{-\omega}(K_{g}-\Delta_g \omega).
\end{equation}
The uniformisation theorem says that in the conformal class $[g]$ of $g$, there exists a unique metric $\hat g=e^{\omega}g$ of scalar curvature $K_{\hat g}=-2$.  Such metrics with uniformized negative curvature  are called hyperbolic.

\subsection*{Green function} 
 Each compact Riemannian surface $(M,g)$ has a   Green function $G_g$ defined 
to be the symmetric integral kernel for the linear operator $R_g: L^2(M)\to L^2(M)$ defined by  $ -\Delta_g R_g f= f$ for any $f$ with zero mean ($\int f \dd {\rm  v}_g=0$)
and $R_g f=0$ for constant functions.
By integral kernel, we mean that for each $f\in L^2(M)$
\[ R_gf(x)=\int_{M} G_g(x,x')f(x'){\rm v}_g(\dd x').\]
 
\begin{lemma}[{\cite[Lemma 2.1]{GRV}}]\label{greenneardiag}
If $g$ is a hyperbolic metric  on the surface $M$, the Green function $G_{g}(x,x')$ for $\Delta_{g}$ has the following form
near the diagonal
\begin{equation}\label{greenfct2} 
G_{g}(x,x')= -\frac{1}{2\pi}\log(d_{g}(x,x'))+m_{g}(x,x')
\end{equation}
for some smooth function $m_g$ on $M\times M$. Near each point $x_0\in M$, there are isothermal coordinates 
$z$ so that ${g}=|dz|^2/{\rm Im}(z)^2$ and near $x_0$ 
\[G_{g}(z,z')= -\frac{1}{2\pi}\log |z-z'|+F(z,z')\]
with $F$ smooth. Finally, if $\hat{g}$ is any metric conformal to ${g}$, \eqref{greenfct2} holds with $\hat{g}$
replacing ${g}$ but with $m_{\hat{g}}$ continuous.
\end{lemma}

 \subsection*{Mabuchi and other classical actions}
Let $M$ be a connected compact surface of genus ${\bf h}\geq 2$ (without boundary).  Let $g\in {\rm Met}(M)$ and $\hat g\in [g]$, i.e. $\hat g=e^{\omega}g$ for some $\omega\in C^\infty(M)$. The K\"ahler potential $\phi:=\phi_{\hat g,g}$ of the metric $\hat g$ 
w.r.t to $g$  is defined by the formula
\begin{equation}\label{kahlerdelta}
\phi=-\tfrac{2}{V_{\hat g}}\int G_{g}(\cdot,y){\rm v}_{\hat g }(\dd y) .
\end{equation}
Another definition of $\phi$ is to define $\phi$ as the unique solution of the equation
$e^{\omega}=\frac{V_{\hat g}}{V_{g}}+\frac{V_{\hat g}}{2}\Delta_{g}\phi$ 
 satisfying $\int \phi \dd {\rm v}_g=0$. Using this potential, the (classical) \textbf{Mabuchi K-energy} can be defined as 
\begin{equation}\label{mabuchi}
S^{\rm cl}_{\rm M}(\hat g,g)=\int_M\Big(2\pi(1-\textbf{h})\phi \Delta_g\phi+(\frac{8\pi(1-\textbf{h})}{V_{g}}-K_{g})\phi+\frac{2}{V_{\hat g}}\omega e^{\omega}\Big) {\rm dv}_{g} .
\end{equation}
We also introduce two classical actions that we use throughout this paper. With the above  notation, the \textbf{Liouville functional},  $S^{{\rm cl},0}_{{\rm L}}(\hat g,g)$, and the  \textbf{Aubin-Yau functional}, $S^{{\rm cl}}_{{\rm AY}}(\hat g,g)$, are respectively defined by
\begin{align}
S^{{\rm cl},0}_{{\rm L}}(\hat g,g):=&\int_M\big(|d\omega|^2_{g}+ 2K_{g}\omega \big){\rm dv}_{g} ,\label{actionL}\\
S^{{\rm cl}}_{{\rm AY}}(\hat g,g):= & \int_M\big(\tfrac{1}{4}\phi\Delta_{g}\phi + \frac{\phi}{V_{g}}\big){\rm dv}_{g}.\label{actionAY}
\end{align}
(This is a slight abuse of notation as $S^{{\rm cl}}_{{\rm AY}}$ is rather a function of $\phi$ than of the metric). For later purposes, notice that both the Mabuchi and Liouville actions satisfy cocycle identities for conformal metrics $g_1,g_2,g_3$
\begin{equation}\label{cocycle}
S^{{\rm cl},0}_{{\rm L}}(g_3,g_1)=S^{{\rm cl},0}_{{\rm L}}(g_3,g_2)+S^{{\rm cl},0}_{{\rm L}}(g_2,g_1)\quad \text{and}\quad S^{\rm cl}_{\rm M}(g_3,g_1)=S^{\rm cl}_{\rm M}(g_3,g_2)+S^{\rm cl}_{\rm M}(g_2,g_1).
\end{equation}
Also, recall the following change of metric formula for Green functions
\begin{equation}\label{rel:green}
G_{\hat g}(x,y)= G_{g}(x,y)+\frac{1}{2}(\phi(x)+\phi(y))- S^{{\rm cl}}_{{\rm AY}}(\hat g,g).
\end{equation}

\subsection*{Regularized determinant of Laplacian} \label{regdet}
Here we summarize results that can be found in \cite{OPS} for instance. For a Riemannian metric $g$ on a connected oriented compact surface $M$, the nonnegative Laplacian $-\Delta_g$ has discrete spectrum
${\rm Sp}(\Delta_g)=(\la_j)_{j\in \N_0}$ with $\la_0=0$ and $\la_j\to +\infty$ sorted in increasing order. We can define the regularized determinant of $\Delta_g$ by 
\begin{equation}
{\det} '(-\Delta_g)=\exp(-\partial_s\zeta(s)|_{s=0})
\end{equation}
where $\zeta(s)$ is the spectral zeta function of $-\Delta_g$ defined as by meromorphic continuation of $\sum_{j=1}^\infty \la_j^{-s}$ which is well defined for ${\rm Re}(s)> 1$ (using Weyl's law $\lambda_j$ is of order $j$).  This extension 
 is defined on the full complex plane and is holomorphic at $s=0$. 
 If $\hat{g}=e^{\omega}g$ for some $\omega\in C^\infty(M)$, the variations of the regularized Laplacian with respect to conformal changes of metrics are determined by the so-called Polyakov formula (see \cite[eq. (1.13)]{OPS}) 
\begin{equation}\label{detpolyakov} 
\log \frac{{\det}'(\Delta_{\hat{g}})}{{\rm V}_{\hat{g}}}= \log \frac{{\det}'(\Delta_g)}{{\rm V}_{g}} -\frac{1}{48\pi}S^{{\rm cl},0}_{{\rm L}}(\hat g,g).
\end{equation}
where $S^{{\rm cl},0}_{{\rm L}}(\hat g,g)$ is the Liouville action \eqref{actionL}.  

\subsection{Gaussian Free Field}\label{sec:GFF}
We refer to   \cite[Section 4.2]{dubedat}   for references concerning this subsection. The Laplacian $-\Delta_g$ has an  orthonormal basis of real valued eigenfunctions $(\varphi_j)_{j\in \N}$ in $L^2(M,g)$ with associated eigenvalues $\la_j\geq 0$ sorted in increasing order. On the Riemannian manifold $(M,g)$, we define the   Sobolev spaces for $s\in \R$ 
$$H^s(M,g):=\Big\{f=\sum_{j\geq 0}f_j\varphi_j\, :\, (f_j)_j\in \R^\N,\,|f|_{H^s}^2:=\sum_{j\ge 1} |f_j|^2\lambda_j^s<\infty\Big\}.$$
Let us denote by $ \cjg \cdot,\cdot \cjd$ the duality bracket. We denote by $H^s_0(M,g)$ the subspace of $H^s(M,g)$ made up of elements $f$ such that $ \cjg f,1 \cjd=0$.

The Gaussian Free Field (GFF) $X_g$ on  $(M,g)$ is a random variable taking values in   $H_0^{-s}(M, g)$ for $s>0$. It is characterized by its mean and covariance kernel for test functions $f,f'\in H^{s}(M,g)$ 
$$\E[ \cjg X_g,f \cjd]=0\quad\text{ and }\quad\E[ \cjg X_g,f \cjd\cjg X_g,f' \cjd]=2\pi \iint_{M^2}f(x)G_g(x,y)f'(y){\rm v}_g(\dd x){\rm v}_g(\dd y).$$
 In view of Lemma \ref{greenneardiag}, the covariance Kernel associated with $X_g$,
 which according to the above equation is given by  $2\pi G_g(x,y)$ displays 
 a pure  logarithmic divergence on the diagonal. With some slight abuse of notation, we use sometimes $\int X_g f{\rm dv}_g$ for $\cjg X_g,f \cjd$.
 \medskip

\subsection{Gaussian multiplicative chaos}\label{GMC}
To make sense of quantities like $e^{\gamma X_g}$ at for some values of $\gamma\in \R$ we use a renormalization procedure after regularization of the field $X_g$. We describe the construction for $g$ hyperbolic and we shall 
remark that in fact the construction works as well for any conformal metric $\hat{g}=e^{\omega}g$ by using Lemma \ref{greenneardiag}.

First, when $\eps>0$ is very small, we define a regularization $X_{g,\eps}$ of $X_g$ by averaging on geodesic circles\footnote{It turns out that other types of regularizations by convolution could work as well but averaging along circles simplifies some computations.} of radius $\eps>0$. Let $x\in M$ and let $\mc{C}(x,\eps)$ be the geodesic circle of center $x$ and radius $\eps>0$, and let $(f^n_{x,\eps})_{n\in \N} \in C^\infty(M)$ be a sequence with $||f^n_{x,\eps}||_{L^1}=1$ 
which is given by $f_{x,\eps}^n=\theta^n(d_g(x,\cdot)/\eps)$ where $\theta^n(r)\in C_c^\infty((0,2))$ non-negative 
supported near $r=1$ such that $f^n_{x,\eps}{\rm dv}_g$ 
is converging in $\mc{D}'(M)$ to the uniform probability measure 
$\mu_{x,\eps}$
on $\mc{C}_g(x,\eps)$ as $n\to \infty$ (for $\epsilon$ small enough, the geodesic circles form a $1d$-manifold and the trace of $g$ along this manifold gives rise to a finite measure, which corresponds to the uniform measure after renormalization so as to have mass $1$, it can also be defined in terms  of $1$-dimensional Hausdorff measure constructed with the volume form on $M$ and restricted to this geodesic circle). Then we state the following result:

\begin{lemma}[{\cite[Lemma 3.2]{GRV}}]\label{Xeps}
The random variable $\cjg X_g,f^n_{x,\eps}\cjd$ converges to a random variable as $n\to \infty$, which has a modification $X_{g,\eps}(x)$ with   continuous sample paths with respect to   $(x,\eps)\in M\times (0,\eps_0)$, with covariance 
\[\mathbb{E}[X_{g,\eps}(x)X_{g,\eps}(x')]=2\pi \int G_g(y,y')d\mu_{x,\eps}(y)\mu_{x',\eps}(\dd y')\]
and we have as $\eps\to 0$
\begin{equation}\label{devptE}
\mathbb{E}[X_{g,\eps}(x)^2]=- \log(\eps)+W_g(x)+o(1)
\end{equation}
where  $W_g$ is the smooth function on $M$ given by $W_g(x)=2\pi m_g(x,x)+ \log(2)$  if $m_g$ is the smooth function of Lemma \ref{greenneardiag}.
\end{lemma}

 Next from Lemma \ref{Xeps}, we are be able to define the Gaussian Multiplicative Chaos (GMC) first considered by Kahane \cite{cf:Kah} in the eighties. An elementary and self-contained construction is presented \cite{berestycki}.

\begin{proposition}[{\cite[Theorem 1.1]{berestycki}}]\label{GMCprop}
 If $\gamma>0$, the random measures $\mc{G}_{g,\eps}^{\gamma}(\dd x):= \eps^{\frac{\gamma^2}{2}}e^{\gamma X_{g,\eps}(x)}{\rm v}_g(\dd x)$ converge in probability and weakly in the space of Radon measures towards a random measure $\mc{G}_g^\gamma(\dd x)$. The measure $\mc{G}_g^\gamma(\dd x)$ is non zero if and only if $\gamma \in (0,2)$.  
\end{proposition}

The construction above is not particular to the hyperbolic metric and works for any field on $\bbR^2$ with logarithmic diverging covariance.
For it to work one just need to show that the divergence of the covariance is not changed after a local isometric mapping of $M$ to $\bbR^2$. This relies only on two facts:
\begin{itemize}
 \item [(i)] The covariance of $X_g$ satisfies 
$2\pi G_g(x,x')=-\log d_g(x,x')+ F(x,x')$ with $F$ continuous.
\item[(ii)] In local isothermal coordinates one can write $g=e^{2f(z)}|dz|^2$ 
\[ \log d_g(z,z')=\log |z-z'|+\mc{O}(1).\]
\end{itemize}

\begin{remark}\label{metricregul}
For later purpose we will need to make the following observation related to changes in the choice of the metric used to regularized the GFF. 
If $\hat{g}=e^\omega g$, consider the GFF $X_g$ with vanishing spatial average in the $g$ metric regularized with circle averages in the $\hat g$ metric
\begin{equation}\label{Xhat}
\hat{X}_{g,\eps}(x):=\lim_{\eps\to 0}\cjg X_g, \hat{f}_{x,\eps}^n\cjd_{\hat{g}}
\end{equation} 
for each $x\in M$ where $\hat{f}_{x,\eps}^n:=\theta^n(d_{\hat{g}}(x,\cdot)/\eps)$ with $\theta^n$ like above, so that $\hat{f}_{x,\eps}^n {\rm dv}_{\hat{g}}$ converge as $n\to \infty$ to the uniform probability measure $\hat{\mu}_{x,\eps}$
on the geodesic circle $\mc{C}_{\hat{g}}(x,\eps)$ of center $x$ and radius $\eps$ with respect to $\hat{g}$.
In isothermal coordinates at $x$ so that some point $z$ in the upper half plane corresponds to the point $x$ and the metric is 
$g=|dz|^2/{\rm Im}(z)^2$, the circle $\mc{C}_{\hat{g}}(x,\eps)$ is parametrized by
\[  \eps e^{-\demi\omega(z)+\eps h_\eps(\alpha)}e^{i\alpha}, \alpha \in [0,2\pi] 
\]
for some continuous function $h_\eps(\alpha)$ uniformly bounded in $\eps$.
Then one has 
\[\mathbb{E}(\hat{X}_{g,\eps}(x)\hat{X}_{g,\eps}(x'))=2\pi \int G_g(y,y')\hat{\mu}_{x,\eps}(\dd y) \hat{\mu}_{x',\eps}(\dd y')\]
and by the arguments in the proof of Lemma \ref{Xeps}, we have as $\eps\to 0$  
\begin{equation}\label{devptE'}
\mathbb{E}(\hat{X}_{g,\eps}(x)^2)=- \log(\eps)+W_g(x)+\demi \omega(x)+o(1).
\end{equation}
 Then by the same arguments as for Proposition \ref{GMCprop}, the random measure 
 \begin{equation}\label{Ggamma}
 \hat{\mc{G}}^\gamma_{g,\eps}(\dd x):= \eps^{\frac{\gamma^2}{2}}e^{\gamma \hat{X}_{g,\eps}(x)}{\rm v}_{\hat{g}}(\dd x)\end{equation}
converges weakly as $\eps\to 0$ to some measure $\hat{\mc{G}}^\gamma_{g}$ which satisfies
\begin{equation}\label{relationentrenorm} 
 \hat{\mc{G}}^\gamma_{g}(\dd x)=e^{\frac{\gamma Q}{2}\omega(x)}\mc{G}^\gamma_{  g}(\dd x).
\end{equation}
\end{remark}

\subsection{Liouville CFT}\label{LCFT}  
Fix  $\gamma\in(0,2)$, $\mu>0$ and set $Q=\tfrac{\gamma}{2}+\frac{2}{\gamma}$. For $F:  H^{-s}(M)\to\R$ (with $s>0$) a bounded continuous functional,   set
\begin{align}\label{partLQFT}
\langle F\rangle_{{\rm L},g}:=& ({\det}'(\Delta_{g})/{\rm Vol}_{g}(M))^{-1/2}  \\
 &\times \int_\R  \E\Big[ F( c+  X_{g}) \exp\Big( -\frac{Q}{4\pi}\int_{M}K_{g}(c+ X_{g} )\,{\rm dv}_{g} - \mu  e^{\gamma c}\mathcal{G}_{g}^\gamma(M)  \Big) \Big]\,\dd c .\nonumber
\end{align}
The expression above gives a mathematical intepretation to the formal functional integral 
\[  \int F(\varphi)e^{-\mathcal{S}_{\rm L}(\varphi,g)}D\varphi\]
where $\mathcal{S}_{\rm L}(\varphi,g)$ is the quantum  Liouville action appearing in\eqref{QLiouville:intro}. Properties true almost surely with respect to the measure $\langle\cdot\rangle_{{\rm L},g}$ will be said true $L$-almost surely.

For the expression in \eqref{partLQFT} to be well defined for every bounded continuous $F$ and that $\langle\cdot\rangle_{{\rm L},g}$ defines indeed a probability measure, one must check that the partition function, that is, the total mass of this measure which is obtained by substituting $F$ by $1$ in the above expression, is finite. 
 
 \begin{theorem}[{\cite[Theorem 1.1]{GRV} (LQFT is a CFT)}]\label{introweyl}
Let $Q=\frac{\gamma}{2}+\frac{2}{\gamma}$ with $\gamma\in (0,2)$ 
and $g$ be a smooth metric on $M$.
For each  bounded continuous functional $F:  H^{-s}(M)\to\R$ (with $s>0$) and each $\omega\in C^\infty(M)$, set $\hat g= e^{\omega}g$. Then
 $\langle F\rangle_{{\rm L}, \hat g} $ is finite  and satisfies the following conformal anomaly:  
\[\langle F\rangle_{{\rm L}, \hat g}= \langle F(\cdot\,-\tfrac{Q}{2}\omega)\rangle_{{\rm L}, g}   \exp\big(\tfrac{1+6Q^2}{96\pi}S_{\rm L}^{{\rm cl},0}(\hat g,g)\big)  .\]
where $S_{\rm L}^{{\rm cl},0}$ is the classical Liouville functional \eqref{actionL}.
Let $g$ be any metric on $M$ and $\psi:M\to M$ be an orientation preserving diffeomorphism, then we have for each bounded measurable $F:H^{-s}(M)\to \R$ with $s>0$
\[ \langle F\rangle_{{\rm L}, \psi^*g}= \langle F(\cdot \circ \psi)\rangle_{{\rm L},  g}  .\]
\end{theorem}

\section{Defining the path integral for Mabuchi+Liouville actions}\label{sec:mabuchi}


In this section can be found the main results of our paper. In particular we first explain the construction of the path integral \eqref{MLintro} and then state its main properties.

Recall that Liouville CFT defines formally a random metric $e^{\gamma (c+X_g)}g$ with the law of the random field $c+X_g$, called Liouville field,  ruled by the path integral \eqref{partLQFT}. The Liouville field being a distribution this metric tensor    is not well defined mathematically. Yet one can make sense of all the corresponding terms appearing in the quantum Mabuchi action with respect to this putative metric tensor because the quantum Mabuchi action only involves the volume form or log-conformal factor, which are well defined under the Liouville path integral.

 Recall the definition \eqref{kahlerdelta} of the K\"ahler potential $\phi$ of some metric $\hat g$ conformal to $g$, i.e. $\hat g=e^{\omega}g$. This expression can be extended to the case when $\hat g$  is the Liouville metric, hence leading to the following expression of the K\"ahler potential of the Liouville metric with respect to a background measure $g$
\begin{equation}\label{kahlerliouville}
 \Phi(x):=-\frac{2}{\mathcal{G}_{g}^\gamma(M)}\int G_g(x,y) \mathcal{G}_{g}^\gamma(\dd y) .
\end{equation}

\begin{proposition}
For $\gamma\in (0,2)$, the K\"ahler potential of the Liouville metric is well defined for  $L$-almost all realization of $X_g$.
It is continuous on $M$ and satisfies $L$-almost surely
$$\int_M\Phi \Delta_{g} \Phi \,d{\rm v}_{g}\leq 0.$$
\end{proposition}

\begin{proof} Multifractal analysis (see \cite[Th 2.14]{review} for instance)  entails that for some constant $C$ and for all $q\in (0,\tfrac{4}{\gamma^2})$
$$\forall x \in M,\forall r\in (0,1),\quad \E[\mathcal{G}_{g}^\gamma(B_g(x,r))]\leq Cr^{\xi(q)}$$
with $$\xi(q)=(2+\tfrac{\gamma^2}{2})q-\tfrac{\gamma^2q^2}{2}.$$
Borel-Cantelli's Lemma then entails uniform H\"older continuity of the measure $\mathcal{G}_{g}^\gamma$,  namely that there exists  $\alpha>0$ and some random variable $C$
$$\forall x \in M,\forall r\in (0,1),\quad \mathcal{G}_{g}^\gamma(B_g(x,r))\leq Cr^{\alpha}.$$
As the   singularity of  Green function is logarithmic (Lemma \ref{greenneardiag}), it is then straightforward  to check that $\Phi$ is well defined (i.e. $\mathcal{G}_{g}^\gamma$ integrates a log-singularity) and is furthermore continuous on $M$. The last statement follows  from the positive definiteness of the Green function.\end{proof}

Our next step is to construct the term corresponding to $\omega e^{\omega}$ in \eqref{introQmabuchi} for the Liouville metric and to check that it satisfies some adequate integrability property. The latter condition requires $\gamma\in (0,1)$ instead of $\gamma<2$ because our proof of existence of exponential moments (Theorem \ref{lefttail} below) includes this restriction. 
The obstruction, however, does not seem to be more than technical and
it seems very plausible that it could be overcome  by refining our technique   to treat the case $\gamma\in [1,2)$. 

\medskip

For $\gep>0$, we consider the random (signed) measures $$\mc{D}_{g,\eps}^{\gamma}(\dd x):= \eps^{\frac{\gamma^2}{2}}(\gamma X_{g,\eps}(x)+\gamma^2\ln \epsilon)e^{\gamma X_{g,\eps}(x)}{\rm v}_g(\dd x).$$
Our result concerns the limit of this object when $\gep$ tends to zero.

\begin{theorem}\label{th:exp}
\begin{itemize}
 \item [(i)] If $\gamma\in (0,\sqrt{2})$ and $f$ is a continuous function $f$ on $M$, the family of random variables
 $\big(\int_M f(x)  \mc{D}_{g,\eps}^{\gamma}(\dd x)\big)_\epsilon$ converges in  quadratic mean towards a limiting random variables denoted   $\int_M f(x)  \mc{D}_{g}^{\gamma}(\dd x)$.
  \item[(ii)] Setting $\mc{D}_{g}^{\gamma}(M):=\int_M  \mc{D}_{g}^{\gamma}(\dd x)$.   Then, for $\gamma<1$ for any $\alpha\geq 0$,
\begin{equation}\label{exp}
\E\Big[\exp\Big(-\alpha \frac{\mc{D}_{g}^{\gamma}(M)}{\mathcal{G}_{g}^\gamma(M)}\Big)\Big]<+\infty.
\end{equation}
\end{itemize}

 \medskip

\end{theorem}

The first part of the theorem results   from an elementary $L^2$-computation (see below).
The proof of the second statement is the main technical part of the paper  and runs from 
Section  \ref{sec:D/M} to \ref{sec:small}. It relies on two distinct concentration estimates for the quantities 
$\mc{D}_{g}^{\gamma}(M)$ and $\mathcal{G}_{g}^\gamma(M)$ which are of independent interest.

\begin{proof}[Proof of Theorem \ref{th:exp} (i)]
 Considering two values $\gep, \gep'>0$, setting $Y_\gep(x):=(X_{g,\eps}(x)+\gamma\ln \epsilon)e^{\gamma X_{g,\eps}(x)}$ we have
 
 \begin{multline}
\gamma^ {-2}\bbE\left[ \Big(\int_M f(x)  \mc{D}_{g,\eps}^{\gamma}(\dd x)- \int_M f(x)  \mc{D}_{g,\eps'}^{\gamma}(\dd x)\Big)^2\right]\\=(\eps\gep')^{\frac{\gamma^2}{2}} \int_{M\times M}   \bbE\left[ 
(Y_{\gep}(x)-Y_{\gep'}(x))(Y_{\gep}(y)-Y_{\gep'}(y))\right]          f(x)f(y)\mathrm v_g(\dd x)\otimes \mathrm v_g(\dd y).
 \end{multline}
Using the fact that the random variables involved are Gaussian we obtain that
 \begin{multline}\label{lanee}
  \bbE[Y_\gep(x)Y_{\gep'}(y)]=e^{\frac{\gamma^2}{2}\bbE[ (X_{g,\eps}(x)+ X_{g,\eps'}(y))^2]}
  \left[  \left( \gamma \bbE[ (X_{g,\eps}(x)+ X_{g,\eps'}(y))X_{g,\eps}(x)]+ \gamma \log \gep\right) \right.\\
  \left.\times \left( \bbE[ (X_{g,\eps}(x)+ X_{g,\eps'}(y))X_{g,\eps'}(y)]+ \gamma \log \gep\right)+ \bbE[ X_{g,\eps}(x)X_{g,\eps'}(x)\right].
  \end{multline}
  In particular we have for $x\ne y$ 
  \begin{multline}
   \lim_{\gep, \gep' \to 0}(\eps\gep')^{\frac{\gamma^2}{2}}\bbE[Y_\gep(x)Y_{\gep'}(y)]\\
   =e^{\frac{\gamma^2}{2} (W_g(x)+W_g(y))+\frac{\gamma^2}{2}G_g(x,y)}\left[(W_g(x)+G_g(x,y))(W_g(y)+G_g(x,y))+G_g(x,y)\right],
  \end{multline}
  and given $\delta>0$ there exists a constant such that 
  \begin{equation}
   (\eps\gep')^{\frac{\gamma^2}{2}}\bbE[Y_\gep(x)Y_{\gep'}(y)]\le C_{\delta} d_g(x,y)^{-(\gamma^2+\delta)}.
  \end{equation}

Using dominated convergence Theorem for the integral in \eqref{lanee} (choosing $\delta$ such that $\gamma^2+\delta<2$), 
the limit of the fourth term in the product cancel out and we conclude that 
$\int_M f(x)  \mc{D}_{g,\eps}^{\gamma}(\dd x)$ is a Cauchy sequence in $L^2$ (for $\gep$ tending to zero). 
\end{proof}

\medskip
\begin{remark} In the same spirit as Remark \eqref{metricregul} we need to study the role of the metric in regularizing the GFF involved in the construction of the random variable $\int_M f(x)  \mc{D}_{g}^{\gamma}(\dd x)$. So we consider another metric $\hat g=e^{\omega}g$ and consider the random variable
\begin{equation}\label{Dgamma}
 \hat{\mc{D}}_{g,\eps}^{\gamma}( \dd x):= \eps^{\frac{\gamma^2}{2}}(\gamma \hat X_{g,\eps}(x)+\gamma^2\ln \epsilon)e^{\gamma \hat X_{g,\eps}(x)}{\rm  v}_g(\dd x)
\end{equation}
as well as the limit $\int f(x)     \hat{\mc{D}}_{g}^{\gamma}(\dd x)=\lim_{\epsilon\to 0}\int_M f(x)  \hat{\mc{D}}_{g,\eps}^{\gamma}(\dd x)$  for continuous functions $f\in C^0(M)$. Similarly to \eqref{relationentrenorm} we obtain the relation  
\begin{equation}\label{transform:mabuchi}
\int f(x)     \hat{\mc{D}}_{g}^{\gamma}(\dd x)=\int_M f(x)e^{\frac{\gamma Q}{2}\omega(x)}    \mc{D}_{g}^{\gamma}(\dd x)+\frac{\gamma^2}{2}\int_Mf(x)\omega(x)e^{\frac{\gamma Q}{2}\omega(x)}    \mc{G}_{g}^{\gamma}(\dd x).
\end{equation}
\end{remark}
 
At this stage, we can provide a mathematical interpretation for the Mabuchi $K$-energy given by Equation \eqref{introQmabuchi} when (formally) $\omega=\gamma X_g$, which is well defined $L$-almost surely. We define the  random variable

\begin{align}\label{potmabLas}
\mathcal{S}_M(c+X_g,g):=&-8\pi(1-\mathbf{h})\frac{1}{ \mathcal{G}_{g}^\gamma(M)^2}\iint_{M^2}G_g(x,x') \mathcal{G}_{g}^\gamma (\dd x)  \mathcal{G}_{g}^\gamma(\dd x')\\
&-\frac{2}{\mathcal{G}_{g}^\gamma(M)} \iint_{M^2} \Big(\frac{8\pi(1-\mathbf{h})}{V_g}-K_g(x)\Big)G_g(x,x')\,{\rm   v}_g(\dd x )\mathcal{G}_{g}^\gamma(\dd x')\nonumber\\
&
+ \frac{2}{1-\frac{\gamma^2}{4} }\frac{1}{\mathcal{G}_{g}^\gamma(M)}\mathcal{D}_{g}^\gamma(M)+ \frac{2}{1-\frac{\gamma^2}{4} }\gamma c\nonumber
 \end{align}
 We can check that each of the three terms above correspond to one term in \eqref{introQmabuchi}.
 For the first one, notice that for a smooth metric $\hat g=e^{\omega}g$ and $\phi$ its K\"ahler potential given by \eqref{kahlerdelta}, we have the relation
 $$\int_M\phi\Delta_g\phi\,{\rm dv}_g=-\frac{4}{V_{\hat g}^2}\iint_{M^2}G_g(x,x'){\rm v}_{\hat g}(\dd x)\otimes {\rm  v}_{\hat g}(\dd x').$$
Replacing ${\rm v}_{\hat g}$ by the volume form  $\mathcal{G}_{g}^\gamma$ gives the first term in \eqref{potmabLas}. The second and third terms in \eqref{potmabLas} then correspond respectively to the second and third term in \eqref{introQmabuchi} in an obvious way given the expression \eqref{kahlerliouville}.

We are now ready to define the QFT associated with the Liouville action and the Mabuchi K-energy similarly to what is done in Theorem \ref{introweyl}.

\begin{defth}{\bf (Quantum  Mabuchi-Liouville)}\label{main}
We fix  $\gamma\in(0,\sqrt{2})$ and $\beta\geq 0$. Let $g$ be a smooth metric on $M$. We define the  Quantum Mabuchi+Liouville theory as the following functional integral: for $F:  H^{-s}(M)\to\R$ (with $s>0$) a bounded continuous functional,  we set
\begin{align}\label{partMLQFT}
\langle F\rangle_{{\rm ML},g}:=& \left(\frac{{\det}'(\Delta_{g})}{{\rm Vol}_{g}(M)}\right)^{-1/2} \\
&\times \int_\R  \E\Big[ F( c+  X_{g}) \exp\Big(- \beta \mathcal{S}_M(c+X_g,g)  -\frac{Q}{4\pi}\int_{M}K_{g}(c+ X_{g} )\,{\rm dv}_{g} - \mu  e^{\gamma c}\mathcal{G}_{g}^\gamma(M)  \Big) \Big]\,\dd c , \nonumber
\end{align}
where the random variable $S_M(c+X_g,g)$ is defined in \eqref{potmabLas}. Furthermore
\begin{description}
\item[1) Finite mass:] The   total mass of this measure, i.e $\langle 1\rangle_{{\rm ML},g}$, is finite for all $\beta\in \big(0,\frac{\mathbf{h}-1}{2}(\tfrac{4}{\gamma^2}-\tfrac{\gamma^2}{4})\big)$.  
\item[2) Weyl anomaly:] Let $\hat g=e^{\omega}g$ be a metric conformal to $g$ and denote by $\phi$ its K\"ahler potential   w.r.t. $g$. Then the ML-path integral obeys the following Weyl anomaly 
\begin{equation}
\langle F\rangle_{{\rm ML},  \hat g}= \langle F(\cdot-\tfrac{Q}{2}\omega)\rangle_{{\rm ML},g}\exp\big( \tfrac{1+6Q^2}{96\pi}S_{\rm L}^{{\rm cl},0}(\hat g,g)+\beta S^{{\rm cl}}_{\rm M}(\hat g, g)\big)
\end{equation}
where $S_{\rm L}^{{\rm cl},0}$ is the classical Liouville functional \eqref{actionL} and $S^{{\rm cl}}_{\rm M}$ the classical Mabuchi action \ref{mabuchi}.
\end{description}
\end{defth} 

\begin{rem}
The expectation in Equation \eqref{partMLQFT} is well defined when $F$ is positive (with $\infty$ being a possible value for $\langle F\rangle_{{\rm ML},  \hat g}$.
For general $F$,  $\langle F\rangle_{{\rm ML},  \hat g}$ is of course properly defined only if  $\langle |F|\rangle_{{\rm ML},  \hat g}<\infty$. Part {\bf 1)} of the statement implies that the integral is well defined for all bounded $F$ when  $\beta\in \big(0,\frac{\mathbf{h}-1}{2}(\tfrac{4}{\gamma^2}-\tfrac{\gamma^2}{4})\big)$.
 \end{rem}   
  
\subsection{Conditioning on area/string susceptibility}\label{sub:string}

Let us explain how we can condition the path integral \eqref{partMLQFT} on having fixed volume. As a consequence, we are able to give the "string susceptibility" (scaling with respect to area) of our QFT. We fix a hyperbolic background metric $g$ on $M$, hence with uniformized scalar curvature. The path integral \eqref{partMLQFT} defines a random geometry with formal metric tensor $e^{\gamma(c+X_g)}g$. The volume form of this metric tensor is thus the random measure 
$$V_\gamma(\dd x):=e^{\gamma c}\mathcal{G}_\gamma(\dd x)$$
with K\"ahler potential \eqref{kahlerliouville}. We are going to compute the law of the couple $(V_\gamma, \Phi)$ conditionally on $V_\gamma(M)=y$ under the probability law defined by the path integral \eqref{partMLQFT}.

For that, define the random variable 
\begin{equation}\label{defP}
P_\gamma:=-8\pi(1-\mathbf{h})\frac{1}{ \mathcal{G}_{g}^\gamma(M)^2}\iint_{M^2}G_g(x,x')d\mathcal{G}_{g}^\gamma(\dd x) \mathcal{G}_{g}^\gamma(\dd x')+ \frac{2}{1-\frac{\gamma^2}{4} }\frac{\mathcal{D}_{g}^\gamma(M)}{\mathcal{G}_{g}^\gamma(M)}.
\end{equation}
and the exponent
\begin{equation}\label{string}
s:=  \frac{2Q}{\gamma}(\mathbf{h}-1) -\frac{2\beta}{1-\frac{\gamma^2}{4}} .
\end{equation}
Let $\mathcal{R}(M)$ be the space of Radon measures on $M$ and $C(M)$ the space of continuous function. We claim

\begin{proposition}{\bf (Fixed volume Quantum  Mabuchi-Liouville Theory)}\label{mainprop}
We fix  $\gamma\in(0,1)$ and $\beta\geq 0$. Let $g$ be a hyperbolic metric on $M$. Let  $F:  \R\times\mathcal{R}(M)\times C(M)\to\R$  a bounded continuous functional,  we have
\begin{multline}\label{partMLQFTUV}
\big\langle F(V_\gamma(M),V_\gamma(\dd x), \Phi) \big\rangle_{{\rm ML},g}\\:=  \gamma^{-1} \left(\frac{{\det}'(\Delta_{g})}{{\rm Vol}_{g}(M)}\right)^{-1/2}  \int_\R  \E\Big[ F\big(y,y\tfrac{\mathcal{G}_\gamma(\dd x)}{\mathcal{G}_\gamma(M)},\Phi\big) \exp(- \beta P_\gamma) \mathcal{G}_\gamma(M)^{-s} \Big]y^{s-1}e^{-\mu y}\,dy. 
\end{multline}
\end{proposition} 
 
\begin{proof}
The relation follows from the simple change of variables $y=  e^{\gamma c}\mathcal{G}_\gamma(M)$ in the $c$-integral in \eqref{partMLQFT}.
\end{proof} 
 
In particular if we define a probability law by dividing the path integral  \eqref{partMLQFT} by its total mass, this shows that the random volume of $V_\gamma(M)$ follows a Gamma law $\Gamma(s,\mu) $. The exponent $s+2$ thus appears as the \textbf{string susceptibility} of our QFT.

\begin{remark}
The loop expansion is the asymptotic expansion as $\kappa^2\to\infty$ (i.e. $\gamma\to 0$) of the string susceptibility in terms of the parameter $\kappa^2=\frac{1+6Q^2}{3}$ or equivalently $\gamma=\big(\frac{3\kappa^2-1}{6}\big)^{1/2}-\big(\frac{3\kappa^2-25}{6}\big)^{1/2}$. Taking our expression \eqref{string} we find at two loops
$$s+2=\tfrac{\kappa^2}{2}(\mathbf{h}-1) +\tfrac{19-7\mathbf{h}}{6}-2\beta+\tfrac{2}{\kappa^2}(\mathbf{h}-1)-\tfrac{4\beta}{\kappa^2} +o(\kappa^{-2}).$$
This is exactly the expression found in \cite[Equation (1.5)]{BFK}.
\end{remark}

\subsection{Proof of Definition-Theorem \ref{main} using Theorem \ref{exp}}

 We start by proving the result when $g=g_0$ is uniformized, i.e. has constant scalar curvature. The reason for that is that the second term in the right-hand side of \eqref{potmabLas} vanishes because of the Gauss-Bonnet theorem \eqref{GB}, which drastically simplifies our task. Indeed, in that case the total mass is then given by
$$\langle 1\rangle_{{\rm ML},g}= \int_{\bbR}\left(\frac{{\det}'(\Delta_{g})}{{\rm Vol}_{g}(M)}\right)^{-1/2}\E\Big[\exp\Big(-\beta P  - \mu  e^{\gamma c}\mathcal{G}_{\hat g}^\gamma(M) \Big)\Big] e^{  - \frac{2\beta\gamma}{1-\frac{\gamma^2}{4} } c-2Q(1-\mathbf{h})c}\,\dd c$$
where $P$ the random variable defined by $$P:=8\pi(\mathbf{h}-1)\frac{1}{ \mathcal{G}_{g}^\gamma(M)^2}\iint_{M^2}G_g(x,x')d\mathcal{G}_{g}^\gamma(\dd x) \mathcal{G}_{g}^\gamma(\dd x')+ \frac{2}{1-\frac{\gamma^2}{4} }\frac{\mathcal{D}_{g}^\gamma(M)}{\mathcal{G}_{g}^\gamma(M)}.$$
The first term in $P$ is negative (due to positive definiteness of $G_g$)
and thus finiteness of the expectation inside the integrale is ensured by the finiteness of negative exponential moments of $(\mathcal{D}_{g}^\gamma/\mathcal{G}_{g}^\gamma)(M)$ (Theorem \ref{exp}).

Now to check that the integral $c$-integral converges when $s:= - \frac{2\beta\gamma}{1-\frac{\gamma^2}{4} }  -2Q(1-\mathbf{h})>0$ (this corresponds to our topological restriction on $\beta$) we observe that a simple change of variables in the $c$-integral 
$c'= \mu \mathcal{G}_{\hat g}^\gamma(M) e^{\gamma c}$ yields
$$\langle 1\rangle_{{\rm ML},g}= \gamma^{-1}\mu^{-\frac{s}{\gamma}}\left(\frac{{\det}'(\Delta_{g})}{{\rm Vol}_{g}(M)}\right)^{-1/2}\E\Big[\exp\Big(-\beta P   \Big)\mathcal{G}_{\hat g}^\gamma(M)^{-\frac{s}{\gamma}}\Big] \Gamma (\frac{s}{\gamma})$$ where $\Gamma$ is the standard Gamma function.

For a generic metric $\hat g$, we use the fact that it is conformal to a uniformized metric $g$, i.e. $\hat g=e^{\omega}g$ for some $\omega\in C^\infty(M)$. Then, for bounded nonnegative functional $F$,
\begin{align*} 
\langle F\rangle_{{\rm ML},\hat g}&=  \left(\frac{{\det}'(\Delta_{\hat g})}{{\rm Vol}_{\hat g}(M)}\right)^{-1/2}  \\
&\times \int_\R  \E\Big[ F( c+  X_{\hat g})  e^{- \beta \mathcal{S}_M(c+X_{\hat g},\hat g) }  \exp\Big(-\frac{Q}{4\pi}\int_{M}K_{\hat g}(c+ X_{\hat g} )\,{\rm dv}_{\hat g} - \mu  e^{\gamma c}\mathcal{G}_{\hat g}^\gamma(M)  \Big) \Big]\,\dd c.
\end{align*}
Now we use the conformal anomaly of the Liouville measure (Theorem \ref{introweyl}) to get
\begin{align} 
\langle F\rangle_{{\rm ML},\hat g}&=  e^{-\frac{1+6Q^2}{96\pi}S_L^0(\hat g,g)}\left(\frac{{\det}'(\Delta_{ g})}{{\rm Vol}_{g}(M)}\right)^{-1/2} \label{computconf} \\
&\times \int_\R  \E\Big[ F( c+  X_{g}-\tfrac{Q}{2})  e^{- \beta \hat{\mathcal{S}}_M }  \exp\Big(-\frac{Q}{4\pi}\int_{M}K_{g}(c+ X_{g} )\,{\rm dv}_{g} - \mu  e^{\gamma c}\mathcal{G}_{g}^\gamma(M)  \Big) \Big]\,\dd c\nonumber
\end{align}
where the random variable $\hat{\mathcal{S}}_M $ is defined as (recall the definitions \eqref{Ggamma} and \eqref{Dgamma})
\begin{align*} 
\hat{\mathcal{S}}_M=& \Big( -8\pi(1-\mathbf{h})\frac{1}{\big(\int_M e^{-\frac{\gamma Q}{2}\omega(x)} \hat{\mathcal{G}}_{g}^\gamma(\dd x)\big)^2}\iint_{M^2}G_{\hat g}(x,x')e^{-\frac{\gamma Q}{2}\omega(x)-\frac{\gamma Q}{2}\omega(x')}\hat{\mathcal{G}}_{g}^\gamma(\dd x) \hat{\mathcal{G}}_{g}^\gamma(\dd x')\\
&-\frac{2}{\int_M e^{-\frac{\gamma Q}{2}\omega(x)}\hat{\mathcal{G}}_{g}^\gamma(\dd x)} \iint_{M^2} \Big(\frac{8\pi(1-\mathbf{h})}{V_{\hat g}}-K_{\hat g}(x)\Big)G_{\hat g}(x,x')e^{-\frac{\gamma Q}{2}\omega(x')}\,{\rm   v}_{\hat g}(\dd x) \hat{\mathcal{G}}_{g}^\gamma(\dd x')\nonumber\\
&
+ \frac{2}{1-\frac{\gamma^2}{4} }\frac{1}{\int_M e^{-\frac{\gamma Q}{2}\omega(x)}\hat{\mathcal{G}}_{g}^\gamma(\dd x)}  \int_M e^{-\frac{\gamma Q}{2}\omega(x)}( \hat{\mathcal{D}}_{g}^\gamma(\dd x)-\tfrac{\gamma Q}{2}\omega(x) \hat{\mathcal{G}}_{g}^\gamma(\dd x))+ \frac{2}{1-\frac{\gamma^2}{4} }\gamma c\Big)\\
=:& A_1+A_2+A_3+ \frac{2}{1-\frac{\gamma^2}{4} }\gamma c.
\end{align*}
Using the relation \eqref{relationentrenorm}   the first in the above right-hand side $A_1$  becomes
$$A_1=-8\pi(1-\mathbf{h})\frac{1}{ \mathcal{G}_{ g}^\gamma(M)^2}\iint_{M^2}G_{\hat g}(x,x')\mathcal{G}_{ g}^\gamma(\dd x) \mathcal{G}_{ g}^\gamma(\dd x').$$
Now the relation \eqref{rel:green} between Green functions leads to the expression
\begin{align}
A_1=& -8\pi(1-\mathbf{h}) \frac{1}{ \mathcal{G}_{g}^\gamma(M)^2}\iint_{M^2}\Big(G_g(x,x')+\tfrac{1}{2}(\phi(x)+\phi(x'))-S^{\rm cl}_{{\rm AY}}(\hat g,g)\Big)\mathcal{G}_{g}^\gamma(\dd x)\mathcal{G}_{g}^\gamma(\dd x')\nonumber\\
=&-8\pi(1-\mathbf{h}) \frac{1}{ \mathcal{G}_{g}^\gamma(M)^2}\iint_{M^2} G_g(x,x')  \mathcal{G}_{g}^\gamma(\dd x) \mathcal{G}_{g}^\gamma(\dd x')+ 8\pi(1-\mathbf{h})S^{\rm cl}_{{\rm AY}}(\hat g,g)\nonumber\\
&-8\pi(1-\mathbf{h}) \frac{1}{ \mathcal{G}_{g}^\gamma(M)}\int_{M} \phi(x) \mathcal{G}_{g}^\gamma (\dd x).\label{anom1}
 \end{align}

Now we focus on the second term. Again,  \eqref{relationentrenorm} produces the first simplification
$$A_2=-\frac{2}{\mathcal{G}_{  g}^\gamma(M)} \iint_{M^2} \Big(\frac{8\pi(1-\mathbf{h})}{V_{\hat g}}-K_{\hat g}(x)\Big)G_{\hat g}(x,x')\,{\rm   v}_{\hat g}(\dd x)\mathcal{G}_{  g}^\gamma(\dd x').$$
Recalling that the ${\rm v}_{\hat g}$ intergral of $G_{\hat g}$ in either variable vanishes, the   contribution of the term $\frac{8\pi(1-\mathbf{h})}{V_{\hat g}}$ reduces to $0$. Then, using the Green relation  \eqref{rel:green} again and the curvature relation \eqref{curvature}, we get
 \begin{align*}
A_2=& \frac{2}{\mathcal{G}_{  g}^\gamma(M)} \iint_{M^2}  K_{\hat g}(x) G_{\hat g}(x,x')\,{\rm   v}_{\hat g}(\dd x)\mathcal{G}_{  g}^\gamma(\dd x')\\
=&\frac{2}{\mathcal{G}_{g}^\gamma(M)} \iint_{M^2}\big(K_g(x)-\Delta_g\omega(x)\big)\Big(G_g(x,x')+\tfrac{1}{2}(\phi(x)+\phi(x'))-S^{\rm cl}_{{\rm AY}}(\hat g,g)\Big)\,{\rm  v}_g(\dd x)\mathcal{G}_{g}^\gamma(\dd x').
 \end{align*}
Then we expand this expression and compute each term, using the  Gauss-Bonnet formula \eqref{GB} when necessary,
\begin{align*}
A_2=&-\frac{2}{\mathcal{G}_{g}^\gamma(M)} \iint_{M^2}\big(\frac{8\pi(1-\mathbf{h})}{V_{  g}}-K_g(x)\big) G_g(x,x') \,{\rm  v}_g(\dd x)\mathcal{G}_{g}^\gamma(\dd x')+\int_MK_g\phi {\rm d v}_g+8\pi(1-\mathbf{h})\frac{1}{\mathcal{G}_{g}^\gamma(M)}\int_M \phi (x) \mathcal{G}_{g}^\gamma (\dd x) \\
&-16\pi(1-\mathbf{h})S^{\rm cl}_{{\rm AY}}(\hat g,g)+\frac{2}{\mathcal{G}_{g}^\gamma(M)}\int_M\omega(x) \mathcal{G}_{g}^\gamma(\dd x) -\frac{2}{V_{\hat g}}\int_M\omega(x)e^{\omega(x)}{\rm   v}_g(\dd x) .
 \end{align*}
 Finally, the third term  can be treated with \eqref{relationentrenorm} and \eqref{transform:mabuchi}
 \begin{equation}\label{A3}
A_3= \frac{2}{1-\frac{\gamma^2}{4} } \frac{\mathcal{D}_{g}^\gamma(M)}{\mathcal{G}_{g}^\gamma(M)}+\frac{\gamma^2-\gamma Q}{(1-\frac{\gamma^2}{4})} \frac{1}{ \mathcal{G}_{g}^\gamma(M)}\int_M\omega(x)\mathcal{G}_{g}^\gamma(\dd x) .
\end{equation}
Combining, we get
$$A_1+A_2+A_3+ \frac{2}{1-\frac{\gamma^2}{4} }\gamma c=\mathcal{S}_M(c+X_g,g)- S_M(\hat g ,g)$$
 
Hence \eqref{computconf} becomes
$$\langle F\rangle_{{\rm ML},  \hat g}=e^{\frac{1+6Q^2}{96\pi}S_L^{{\rm cl},0}(\hat g,g)+\beta S^{\rm cl}_M(\hat g, g)}\langle F(\cdot-\tfrac{Q}{2}\omega)\rangle_{{\rm ML},g}$$
when $g$ is uniformized. Taking $F=1$ this relation shows that the conditions for finiteness of the total mass does not depend on the choice of the background metric $\hat g$.  The general case ($g$ not necessarily uniformized) then results from the above relation and cocycle identities \eqref{cocycle} for (classical) Liouville and Mabuchi functionals.\qed
\section{Proof of Theorem \ref{th:exp} (ii): Negative exponential moments for GFF based $\mathcal{D}/\mathcal{G}$}\label{sec:D/M}

In this section we first state   general results for GMC/DGMC based on a field defined in the plane and satisfying some assumptions including the existence of a continuous scale decomposition. Our second task is to show that GFF on a manifold can be remapped to the plane using local charts in order to fit those assumptions.

\subsection{Setup and result}\label{qatrin}

In this section we consider a distributional Gaussian field $X$  with covariance function $K$ defined over an open neighborhood of the closure of a bounded open set $D$, and we let $\mu$ be a finite mass Borel measure on $D$.

\begin{assumption}{\bf (Smooth white noise decomposition)}  \label{ass:WN}

  \begin{enumerate} 
 \item The covariance covariance kernel $K$ can be written in the form 
\begin{equation}\label{inzeform}
 K(x,y):= \int_{0}^{\infty} Q_u(x,y) \dd u,\
 \end{equation}
 where the above integral is convergent for all $x\ne y$ and
  $Q_u$ is a bounded symmetric positive definite kernel for any fixed $u$.
 
\item  The function $(x,y)\mapsto K(x,y)-\log |y-x|$ can be extended on the diagonal to a bounded continuous function on $D\times D$. Setting $K_t:= \int^t_0 Q_u \dd u$. There exists a  positive constant $C$ such that 
 \begin{equation}\label{asymptex}
\left | K_t(x,y)-  (t \wedge \ln_+ \tfrac{1}{|x-y|}) \right| \le C.
\end{equation}
\item  We have $\lim_{x\to \infty}Q_u(x,x)= 1$ with uniform convergence in $x\in D$ . 
\item For all $0<\alpha<2$ , $\iint_{D^2}\int_{0}^{\infty} e^{\alpha u}|Q_u(x,y)|\,\mu(\dd x)\mu( \dd y)\dd u<\infty$.
\item We have $\mu(\dd x)=h(x)\dd x$ for some positive bounded continuous function $h$, and 
there exists $\upsilon>1$
\begin{equation}\label{intfunction}
 \int^{\infty}_t  |Q_u(x,y)|\,\dd u\leq C e^{-e^{\upsilon t}|x-y|^{\upsilon}}.
\end{equation}
 \end{enumerate}
 \end{assumption}

For our application we  consider $\mu(\dd x)= h(x)\dd x$ where $h$ is a smooth positive function.
Note that the assumption implies in particular that $K$  has logarithmic divergence on the diagonal ($K(x,y)= -\log |x-y| +O(1)$) so that the GMC and its derivative (cf. Proposition \ref{GMCprop} and Theorem \ref{th:exp} (i)), can properly be defined.

Letting $X_{\gep}$ denote the circle average of the field $X$ defined by 
$X_\gep(x):=\int X_{\gep}(x-y) \hat \mu_{\gep}(\dd y)$ where $\hat \mu_{\gep}$ is the uniform measure on the circle of radius $\gep$ 
(the convolution with a singular measure can be obtained by functional approximation like in Lemma \ref{Xeps}). Given $w$ a positive continuous function on $\bar D$ (in particular $w$ is bounded away from $0$ and $\infty$), we define \
\begin{equation}\label{lacircledef}
\begin{split}
\mathcal{G}^w_\infty&
:= \lim_{\gep\to 0} \int_{D}e^{\gamma X_{\gep}-\frac{\gamma^2}{2}\left( \bbE\left[ X_{\gep}^2 \right]-w(x)\right)} \dd \mu, \\
\mathcal{D}^w_\infty&:= \lim_{\gep\to 0} \int_{D}\left[\gamma X_{\gep}-\gamma^2 \left(\bbE\left[ X_{\gep}^2 \right]-w(x)\right)\right]e^{\gamma X_{\gep}-\frac{\gamma^2}{2} \left(\bbE\left[ X_{\gep}^2 \right]-w(x)\right)} \dd \mu.
\end{split}
\end{equation}

\begin{remark}
Note that our convention here for renormalization   is a bit different of the convention adopted in Section \ref{GMC} since $\log \gep$ has been replaced by $\bbE\left[ X_{\gep}^2 \right]$. Yet, according to item 2 of Assumption \ref{ass:WN}, which implies $\lim_{\epsilon\to 0}\bbE\left[ X_{\gep}(x)^2 \right]+\log \gep$ is a continuous function, 
 the effect of this change on $\cG$ is to multiply the integration measure $\mu$ by a multiplicative factor. Tuning $w$ accordingly cancels this difference (same occurs for $\cD$). 
\end{remark}

We are going to prove that Theorem \ref{th:exp} (ii) is a consequence of the following general statement.
   
\begin{proposition}\label{momexpmean}
If Assumption \ref{ass:WN} holds and $\gamma<1$ then for any  $\alpha>0$ and any $w$  we have
\begin{equation}\label{momentexpB}
\E\Big[e^{-\alpha \frac{\mathcal{D}^w_\infty}{\mathcal{G}^w_\infty}}\Big]<\infty.
\end{equation}
\end{proposition}
  
  The proposition above is proved by combining two new results concerning Gaussian Multiplicative Chaos whose proof are given in Section \ref{sec:D} and \ref{sec:small} respectively. In the two next statements $\cD_{\infty}$ and $\cG_\infty$ stand for
  the limits obtained in the case $w=0$.
  The first one entails that $\cD_{\infty}$ has a subgaussian negative tail.

  \begin{theorem}\label{lefttail}
If Assumption \ref{ass:WN} (1-4) holds then for $\gamma < 1$, there exists a constant $c(D,\gamma, Q)>0$ such that for any $v>0$
\begin{equation}
\bbP\left[ \cD_{\infty}<-v \right] \le  2 e^{-c v^2}.
\end{equation}  
\end{theorem}

The second one asserts that the probability of $\cG_{\infty}$ being smaller than $s$ is subexponential in $1/s$ if the average value of the field is subtracted in the chaos expression. More precisely if we set 
$$m_h(X):= \frac{1}{\mu(D)}\int_D X  \mu(\dd x),$$ we have the following.

\begin{theorem}\label{laprop}
 Given a field satisfying Assumption \ref{ass:WN}, then for all $\gamma\in (0,2)$, for  $\zeta$  sufficiently large 
 there exists a constant $c>0$ such that for all $s>0$ an
 \begin{equation}\label{lapropas}
 \P \left (  e^{-\gamma m_h(X)} \cG_{\infty} \le  s  \right )  \leq 2 e^{-c |\log s|^{-\zeta}\, s^{-\frac{4}{\gamma^2}}}.
 \end{equation}
 \end{theorem}

 \begin{remark}In the presentation of the result, we have assumed that the measure $\mu$ in the definition of $m_\mu(X)$ and the one used for the GMC \eqref{lacircledef} are the same, but this assumption can be relaxed. For instance the result still holds true with $m_h(X)$ replaced by 
 $m_{h'}(X):= \frac{1}{\int_D h'(x)\dd x}\int_D X  h'(x)\dd x,$ for any bounded continuous function $h'$ (to see this it is sufficient to check that the GMC associated with $h'$ is larger than $\frac{\min |h'|}{\max |h|}\cG_{\infty}$).
 
 Moreover most of the proof could work on more general assumption on the measure $\mu$ and one could almost replace $h(x)\dd x$ by a much more singular measure,
 say concentrated on a fractal set (the exponent $\frac{4}{\gamma^2}$ though would be altered and depend on the fractal dimension of the measure). The only part of our proof which uses the fact that the measure is nice is estimate of the $L^p$ moment in Lemma \ref{lpmom}.  Let us mention that the problem of small deviations for GMC have been the object of particular interest recently, with \cite{GHSS} investigating the problem on fractal sets. 
 \end{remark}

\begin{proof}[Proof of Proposition \ref{momexpmean} from Theorem \ref{lefttail} and \ref{laprop}]
For readability we use the notation $Z=m_h(X)$ and assume first that $w\equiv 0$.
Let us consider 
$$q(x):= -2\int_{D} K(x,y) \mu(\dd x)+\int_{D} K(x,y) \mu(\dd x)\otimes
\mu(\dd y)= -2\bbE[X(x)Z]+E[Z^2]$$
We observe that  $e^{-\gamma Z} \cG_{\infty}$ is the GMC associated  with the field $X-Z$ and integrated w.r.t to measure $\tilde \mu(\dd x):=e^{\frac{\gamma^2 q(x)}{2}}\mu(\dd x)$. 
The covariance of $(X-Z)$ can be written in the form
$$\tilde K(x,y)= \int_{0}^{\infty} \tilde Q_u(x,y) \dd u$$
where, setting $\bar h=(\mu(D))^{-1} h$ as the renormalized density, 
\begin{equation}\label{deftildeQ}
\tilde Q_u(x,y):= Q_u(x,y)-\int_D \left( Q_u(x,z)+ Q_u(y,z)\right)\bar h(z)  \dd z
\int_D Q_u(z,z')\bar h(z) \bar h(z')\dd z\dd z'.
\end{equation}
It is a tedious but straightforward computation to check that $\tilde Q$ 
 (and $\tilde \mu$) satisfy Assumption \ref{ass:WN}.
The associated derivative GMC is given by 
\begin{equation}
\tilde{\mathcal D}_{\infty}=
e^{-\gamma Z} \cD_{\infty}-  \gamma Ze^{-\gamma Z}\cG_{\infty} 
- \gamma^2 e^{-\gamma Z} \cG^q_{\infty},
\end{equation}
where $\cG^q_{\infty}$ is the GMC associated with $X$ and the (non-necessary positive)
measure $q(x)\mu(\dd x)$.
As a consequence of the definition of GMC we have $|\cG^q_{\infty}/\cG_{\infty}|\le \| q\|_{\infty},$ and thus using Cauchy-Schwartz, we have
 
 \begin{multline}
 \E\Big[e^{-\alpha \frac{\mathcal{D}_\infty}{\mathcal{G}_\infty}}\Big]
 \le  e^{\alpha \gamma^2\| q\|_{\infty}} \E\Big[\exp\Big(-\frac{\mathcal{\alpha \tilde D}_\infty}{e^{-\gamma Z}\mathcal{G}_\infty}-\alpha\gamma Z\Big)\Big]\\
 \le  e^{\alpha \gamma^2\| q\|_{\infty}} \left(\E\Big[\exp\Big(-\frac{2\alpha \tilde{\mathcal{D}}_\infty}{e^{-\gamma Z}\mathcal{G}_\infty}\Big)\Big] \bbE[ e^{-2\alpha\gamma Z}]\right)^{1/2},
 \end{multline}
and we only need to show that the first factor in the square root is finite.
Now we use the fact if $a>0$ and $b\in \bbR$ we have 
$$ ab \le \frac{1}{3}a^{3}+ \frac{2}{3}(b_+)^{3/2}$$
 for  $a=  e^{\gamma Z}(\cG_{\infty})^{-1}$ and $b=-\tilde\cD_{\infty}$.
 We obtain using Cauchy-Schwarz again
 \begin{multline}
  \E\Big[\exp\Big(-\frac{2\alpha \tilde{\mathcal{D}}_\infty}{e^{-\gamma Z}\mathcal{G}_\infty}\Big)\Big]^2\le \bE\Big[e^{ 2\alpha( (e^{-\gamma Z}\mathcal{G}_\infty)^{-3}+ 
  |\mathcal{D}_\infty|^{3/2} \ind_{\{\cD_{\infty}<0\}})}\Big]^2\\
  \le \bE\Big[e^{4\alpha(e^{-\gamma Z}\mathcal{G}_\infty)^{-3}}\Big] \bE \Big[ e^{4\alpha|\mathcal{D}_\infty|^{3/2} \ind_{\{\cD_{\infty}<0\}}}\Big].
 \end{multline}
And finally we can use Theorems \ref{laprop} and \ref{lefttail} (recall that $\gamma<1$)
to conclude. The proof remains valid if we have to consider $w+q$ rather than $q$ for $\tilde \mu$, which yields the general statement.
\end{proof} 

 \subsection{Fitting the GFF on $M$ to Assumption \ref{ass:WN}} 

To prove  Theorem \ref{th:exp}(ii) from Proposition \ref{momexpmean}, 
we need to use the local charts to map the field onto some domain of the plane, 
and prove the right estimate for the corresponding covariance function.
Thus we  first show that we can restrict   to prove Equation \ref{exp} only for GMC associated with local neighborhoods of $M$.
 To control the covariance function on this neighborhoods after mapping them to the complex plane, we compare it to that of the Dirichlet GFF on a disk.

\medskip

We consider the Gaussian Free Field $X_g$ on a compact Riemann surface $M$ without boundary endowed with a metric $g$. Given $\delta>0$, we let $\cB(\delta,M)$ be the set of smooth domains of small diameter on $M$

\medskip

$$\cB(\delta,M):=\{ S\subset M \ : \ S \text{ open, } \partial S \text{ is a piecewise $C^1$ Jordan curves and } {\rm diam}_g(S)\le \delta\}. $$

\begin{lemma}\label{lem:expJ} 
There exists $\delta>0$ such thar have for all $S\in \cB(\delta,M)$ for all $\alpha>0$
\begin{equation}\label{momentexp}
\E\big[e^{-\alpha \mathcal{D}^\gamma_g(S)/\mathcal{G}^\gamma_g(S)}\big]<\infty.
\end{equation}
\end{lemma}

\begin{proof}[Proof of Theorem \ref{th:exp}(ii) from Lemma \ref{lem:expJ}] Given a subset $A\subset M$ ,we denote by $\bar{A}$ its closure. As $M$ is a compact manifold, it can be covered by a finite union $\bigcup_{i\in I}\bar{S}_i$, where  $S_i\in \cB(\delta,M)$ for all $i$ together $\bigcap_{i\in I}S_i=\emptyset$. As almost surely  $\mathcal{D}^\gamma_g(\bar{S}_i)=\mathcal{D}^\gamma_g(S_i)$ and $\mathcal{G}^\gamma_g(\bar{S}_i)=\mathcal{G}^\gamma_g(S_i)$  for all $i$, we have
\begin{align*}
\E\big[e^{-\alpha \mathcal{D}^\gamma_g(M)/\mathcal{G}^\gamma_g(M)}\big]=\E\Big[e^{-\alpha  \sum_{i\in I}\lambda_i\frac{\mathcal{D}^\gamma_g(S_i)}{ \mathcal{G}^\gamma_g(S_i)} }\Big]
\end{align*}
with $\lambda_i=\frac{ \mathcal{G}^\gamma_g(S_i)}{ \mathcal{G}^\gamma_g(M)}$. Observe that the $\lambda_i$'s are positive and sum up to $1$. By Jensen,
\begin{align*}
\E\big[e^{-\alpha \mathcal{D}^\gamma_g(M)/\mathcal{G}^\gamma_g(M)}\big]\leq \E\Big[\sum_{i\in I}\lambda_ie^{-\alpha  \frac{\mathcal{D}^\gamma_g(S_i)}{ \mathcal{G}^\gamma_g(S_i)} }\Big]\leq \sum_{i\in I}\E\Big[  e^{-\alpha  \frac{\mathcal{D}^\gamma_g(S_i)}{ \mathcal{G}^\gamma_g(S_i)} }\Big]<+\infty.
\end{align*}
Hence our claim.
\end{proof} 

Our next step is to now to replace $S\in \cB(\delta,M)$ by a subset of $\bbC$ using local charts.
By a compactness argument and choosing $\delta$ small enough,  
we can assume that every $S\in \cB(\delta,M)$ can be included within the image of an open disc (or two dimensional ball) for a given chart.
More precisely we can assume
$\bar{S}\subset \psi(B)$ with ${\rm dist}(\bar S,\psi(B)^c)>0$,  where $B$ is an open ball with an isothermal coordinate  chart $\psi:B\subset \C\to M$.
It is then equivalent to prove \eqref{momentexp}
for the GMC (and derivative) associated with the  Gaussian field $X_g\circ \psi$ on $B$ endowed with a conformal metric $g_\psi(z)|dz|^2$ for some function  
$g_\psi(z):=e^{\omega(z)}$ with $\omega$ smooth, integrated over the set $\psi^{-1}(S)$. Note that without lack of generality one can assume that $B$ is the unit disc centered at $0$.

\medskip

The next step is to write $X_g\circ \psi$ as the sum of a Dirichlet Gaussian Free Field (see appendix \ref{diriche} for the definition).

\medskip
The first correction term we want to discard is the spatial average of the field.
We let $\hat \mu_{\partial B}$ denote the uniform probability measure over the circle $\partial B$ and define $m_B:= \int_{\partial B}X_g\circ \psi \,{\rm d}\hat \mu_{\partial B}$ the spatial average  of the field over the circle $\partial B$. 

\medskip

We set  $\tilde X_B:=X_g\circ \psi-m_B$. Let us define $\mathcal{D}_B(\dd x)$ and $\mathcal{G}_B(\dd x)$ using the same procedure as
 in Theorem \ref{exp} and Proposition \ref{GMCprop} but replacing $X_g$ by $\tilde X_B$ and $\mathrm v_g$  by $g_\psi(z)|dz|^2$. 
One can easily check that
\begin{equation*}
 \mathcal{D}_B(\psi^{-1}(S))=e^{-\gamma m_B}\big(\mathcal{D}_g^\gamma(S)-\gamma m_B\mathcal{G}_g^\gamma(S)\big) \quad \text{ and } \quad  \mathcal{G}_B(\psi^{-1}(S))=e^{-\gamma m_B}\mathcal{G}_g^\gamma(S),
\end{equation*}
and hence 
\begin{equation*}
\frac{\mathcal{D}^\gamma_g(S)}{\mathcal{G}^\gamma_g(S)}=\frac{\mathcal{D}_B(\psi^{-1}(S))}{\mathcal{G}_B(\psi^{-1}(S))}+\gamma  m_B.
\end{equation*}
We have thus for any $\alpha>0$ by Cauchy-Schwartz inequality 
\begin{equation}\label{momentexpBS}
\E\big[e^{-\alpha \mathcal{D}^\gamma_g(S)/\mathcal{G}^\gamma_g(S)}\big]\le \E\big[e^{-2\alpha \mathcal{D}_B(\psi^{-1}(S))/\mathcal{G}_B(\psi^{-1}(S))}\big]^{1/2}  \E[e^{-2\alpha\gamma m_B}]^{1/2}.
\end{equation}
The variable $m_B$ being Gaussian, the second factor is finite for every $\alpha>0$.   The proof of Lemma \ref{lem:expJ} is reduced to a statement about fields 
defined on subsets of the unit ball whose closure does not intersect the boundary
(we need to prove \eqref{momentexp} with $\mathcal{D}^\gamma_g(S)/\mathcal{G}^\gamma_g(S)$  replaced by  $\mathcal{D}_B(\psi^{-1}(S))/\mathcal{G}_B(\psi^{-1}(S))$). The situation is more comfortable with $\tilde X_B$ than with the field $X_g\circ \psi$ (which is defined on the same set) because  the domain Markov property (see Appendix \ref{diriche}) ensures that it can be written as an independent sum    $X+H$ where $X$ is a Dirichlet GFF inside $B$ and $H$ is the harmonic extension of the boundary values of $\tilde X_B$.

\medskip

In the remainder of the proof, we must show that $\mathcal{D}_B(\psi^{-1}(S))$
and $\mathcal{G}_B(\psi^{-1}(S))$ are obtained as limits in \eqref{lacircledef}, and that the  covariance of $\tilde X_B$ satisfies Assumption \ref{ass:WN}.

\medskip

 For the first point, it can be checked via second moment computation that taking
average over Euclidean circles  and subtracting a variance term (like in \eqref{lacircledef}), or taking average over
 circles induced by $g_\psi(z)|dz|^2$ and using  the renormalization convention of Section \ref{GMC} which amounts to subtract a $\log \gep$ term (which is  what is done for $\mathcal{D}_B(\psi^{-1}(S))$ and $\mathcal{G}_B(\psi^{-1}(S))$, amounts to the same result provided that one chooses $w=W_g\circ \psi^{-1}$ (recall \eqref{devptE}).

\medskip

 Concerning the second point, if we let $G_B$ denote the Dirichlet Green function on the unit disc (let us stress that $G_B$ does not depend on the tensor $g_{\psi}(z)$ and is thus completely explicit) and $K_H$ the covariance of the field $H$, the covariance of $\tilde X_B$ is given by $2\pi G_B+K_H$. 
 
 If we let $p_B$ denote the heat-Kernel associated with the (flat) Dirichlet Laplacian on the unit ball we have
\begin{equation}\label{bounGreen}
G_B (x,x')=   \int_{0}^{\infty}p_B(t,x,x')\dd t.
 \end{equation}
To fit our decomposition, and setting $$\bar G(x,x'):= \int_{1}^{\infty}p_B(t,x,x')\dd t,$$ we can rewrite it as  
 \begin{equation}\label{bounGreen2}
G_B (x,x')=  2\int_{0}^{\infty}p_B(e^{-2u},x,x')e^{-2u}  \dd u+ \bar G(x,x').
 \end{equation}
 Now we have 
 \begin{equation}
  2\pi G_B+K_H=\int^\infty_0 Q_u \dd u,
 \end{equation}

where 
$$Q_u(x,x')= 4\pi e^{-2u} p_B(e^{-2u},x,x') + (K_H+2\pi\bar G)(x,x')\ind_{\{u\in [0,1]\}}.$$
  To check that $Q$ satisfies Assumption  \ref{ass:WN} we only need to care about the first term (as the second one is clearly bounded and measurable).
 Now the heat-kernel on the unit ball can be expressed as  
 $$p_B(t,x,y)= \frac{1}{4\pi t}e^{-\frac{|x-y|^2}{4t}}-R(t,x,y).$$
where the first term corresponds to the heat-kernel on the full plane and $R(t,x,y)\ge 0$ corresponds to a correction term which accounts for the fact that the diffusion is killed at the boundary and is small when $t$ is small and $x$ and $y$ are away from the boundary.
Note that $Q^0_u(x,y):=e^{-\left(\frac{|x-y|}{2 e^u}\right)^2}$ satisfies Assumption \ref{ass:WN} trivially and we must thus only check that the correction term induced by $R$ is not relevant.
It is a classical estimate (see e.g \cite[section D.2]{Rnew12})  that given $\delta>0$, there exists a constant $c>0$ such that for all $x, y \in B(0,1-\delta)$ and $t\le 1$ 
\begin{equation}
R(t,x,y)\leq \frac{1}{c}e^{-c/t}.\label{R2}
\end{equation}
As a consequence we have for $u\ge 1$, $x,y \in \psi^{-1}(S)\subset  B(0,1-\delta)$,  
$$|Q_u(x,y)-Q^0_u(x,y)|= 4\pi e^{-2u}R(e^{-2u},x,y)\le  e^{-c e^{2u}},$$
which is sufficient to prove that $Q_u$ also satisfies  Assumption \ref{ass:WN}.
 
\section{Exponential moments for the DGMC: proof of Theorem \ref{lefttail}}\label{sec:D}
   
 \subsection{Brownian decomposition and martingales}\label{bdec}
 
 Let us first explain the importance of the integral representation of the covariance function $K$. As the result depends only on the distribution on $X$ we may construct the process as we wish. We choose to think of it as a limit of a continuous martingale.
We define $(X_t(x))_{x\in D, t\ge 0}$ to be the jointly continuous process in $x$ and $t$  with covariance kernel is given by 
\begin{equation}\label{inzeform}
 \E[X_t(x)X_s(y)]:= \int_{0}^{t\wedge s} Q_u(x,y) \dd u.\
 \end{equation}
Note that given $x\in D$, the process $(X_t(x))_{t\ge 0}$ is a Brownian Motion 
with a deterministic time change given by $K_t(x,x)=\int_{0}^{t} Q_u(x,x)\dd t$. According to our assumptions \eqref{asymptex} and $Q_u\le 1$ we have 
\begin{equation}\label{timechange}
 |K_t(x,x)-t|\le C \quad \text{ and }  \lim_{t\to \infty}\sup_{x\in D}|\partial_t K_t(x,x)-1|=0,
\end{equation}
which makes the process very similar to a standard Brownian Motion.
For $\gamma \in (0,2)$, we define the random measure on $D$ 
\begin{equation}\label{defgt}
\mathcal{G}_t( \dd x):=  e^{\gamma X_t(x)-\frac{\gamma^2}{2} K_t(x,x)}\,  \mu(\dd x)
\end{equation}
 and for $\gamma \in (0,\sqrt{2})$the random distribution on $D$
\begin{equation}
\mathcal{D}_t( \dd x):= (X_t(x)-\gamma t)e^{\gamma X_t(x)-\frac{\gamma^2}{2} K_t(x,x)} \, \mu(\dd x).
\end{equation}
With some harmless abuse of notation we set
 $\mathcal{G}_t:= \int_D\mathcal{G}_t(\dd x)$ 
and $\mathcal{D}_t:= \int_D  \mathcal{D}_t(\dd x)$.
 By construction $(\mathcal{D}_t)_{t\ge 0}$ is a martingale for the filtration $(\cF_t)_{t\ge 0}$.
An explicit computation of the variance shows that $\mathcal{D}_t$ is uniformly bounded in $L^2$ when $\gamma^2<2$. 
Let us call $\mathcal{D}_\infty$ and $\mathcal{G}_\infty$ the limit. 
Standard $L^2$ computations (similar the one performed in the proof of Theorem \ref{th:exp}) allow us to show that these limits coincide with that obtained in \ref{lacircledef} when $\gamma<\sqrt{2}$. 
In this framework we can use stochastic calculus to prove concentration-type results for $\cD_\infty$.

\subsection{Decomposition of the proof}

We assume, for better readability 
that  $Q_u(x,x)=1$ so that $K_t(x,x)=t$ for all $t$. The reader can check that all the computation can be adapted when we only have
\eqref{timechange}.

The core idea of the proof is to obtain a bound on the negative exponential moment of $\cD_{\infty}$ by using its predictable bracket (which for continuous martingales 
coincides with the quadratic variation).
Indeed if $(M_t)_{t\ge 0}$ is a continuous martingale with initial condition $0$, we have
\begin{equation}\label{expomar}
 \bbE\left[e^{-\alpha M_t -\frac{\alpha^2}{2}\langle M \rangle_t}\right]=1,
 \end{equation}
as the expression inside the expectation is also a martingale. 
This implies that if  $\langle M \rangle_\infty=\lim \langle M \rangle_t$ is uniformly bounded ($\esssup  \langle M \rangle_\infty<\infty$) then 
the limit of $M_t$ displays Gaussian concentration.

\medskip

Note that there is no hope to prove directly that $\langle \mathcal{D}\rangle_\infty$ is uniformly bounded: indeed $\mathcal{D}_\infty$
does not display Gaussian concentration as we expect the same right tail as GMC, i.e. 
 $ \bbE\left[ |\mathcal{D}_\infty|^p \right]=\infty$ whenever $p >\frac{4}{\gamma^2}$, see \cite{RVtail}.

\medskip

An explicit computation of $\mathcal{D}_t$ can however give us some extra intuition on the problem.
Setting 
$$W^0_t(x):= (X_t(x)-\gamma t)e^{\gamma X_t(x)-\frac{\gamma^2}{2}t},$$
which is a martingale in $t$ for all $x$, 
and following the rule of Itô calculus we obtain the following expression 
\begin{multline} \label{braketail}
\langle \mathcal{D}\rangle_t= \int_{[0,t]\times D^2}  \dd \langle W^0(x), W^0(y) \rangle_u  \mu(\dd x) \mu( \dd y)\\
=\int_{[0,t]\times D^2}Q_u(x,y) 
\left(1+ \gamma(X_u(x)-\gamma u)\right)\left(1+ \gamma(X_u(y)-\gamma u)\right) e^{\gamma X_t(x)}e^{\gamma X_u(x)+X_u(y)-\gamma^2 u} \,\dd u.
\end{multline}
We see from the above expression that most of the contribution to $\langle \mathcal{D}\rangle_t$ is given by high values of $X_u$, $u\in[0,t]$.
However, these high-values  should not contribute much to the negative tail of $\cD_t$ since they tend to yield high positive values of the integrand
$W^0_t(x)$.

\medskip

Hence our idea is to compare $\cD_t$ with a martingale obtained by replacing $W^0_t(x)$ by an alternative martingale $W_t$
which does not sum the variation of $W^0_t$ when the value of $X_t$ is to large.  
Our definition of $W_t$ has to be carefully chosen so that it compares well with $\cD_t$.
Let us fix $A>0$ sufficiently large and $\eta$ which satisfies
\begin{equation}\label{asumpeta}
 2\eta> \gamma \quad \text{ and } \gamma^2+2\eta\gamma<2. 
\end{equation}
(Note here that with our assumption $\gamma<1$, $\eta=1/2$ satisfies \eqref{asumpeta}, we felt however that 
keeping a parameter would make the computation more readable. The choice $A=100$ is also amply sufficient for all our computation).

Then we define the following sequence of stopping times (with the convention that $R_0^x:=0$) for $k\geq 1$
\begin{align}
T^x_k:=&\inf\{t\geq R_{k-1}^x \ : \  X_t(x)=(\gamma+\eta)t+A\}\\
R^x_k:=&\inf\{t\geq T_{k}^x \ : \ X_t(x)=\gamma t\}.
\end{align}
Introducing the notation $\mathcal{R}^x:=\bigcup_{k=1}^\infty[R_{k-1}^x,T_{k}^x],$ we define (the second equality being derived from  It\^o's formula)
\begin{equation}
W_t(x):=\int_{[0,t]\cap \mathcal R^x } \dd W_t(x)= 
\int_0^t\Big(\gamma(X_s(x)-\gamma s)+1\Big)e^{\gamma X_s(x)-\frac{\gamma^2}{2}s}\mathbf{1}_{\{s\in \mathcal{R}^x\}}dX_s(x)
\end{equation}
Now we reader can check the correctness of the following alternative expression for $W_t(x)$ 
the fact that $W^0_t(x)$ cancels at times $R^x_i$,
\begin{equation}\label{yop}
W_t(x)=\left\{
\begin{array}{ll}
 (X_t(x)-\gamma t)e^{\gamma X_t(x)-\frac{\gamma^2}{2}t}   +\sum_{i=1}^{k}\big(A+\eta   T^x_i\big)e^{\gamma A+(\gamma \eta+\frac{\gamma^2}{2})  T^x_i}   & \text{if }t\in (R^x_{k},T^x_{k+1})\\
 &\\
 \sum_{i=1}^k\big(A+\eta  T^x_i\big)e^{\gamma A+(\gamma \eta+\frac{\gamma^2}{2}) T^x_i}   & \text{if }t\in (T^x_{k},R^x_{k})
\end{array}
\right.
\end{equation}
By construction $(W_t(x))_{t\ge 0}$ is a $\mathcal{F}_t$-martingale and so is
 \begin{equation}
\tilde{\mathcal{D}}_t:=\int_D W_t(x)\, \mu(\dd x).
\end{equation}
Furthermore, repeating the computation from \eqref{braketail},  
we obtain
\begin{equation}\label{brackettilde}
\langle \tilde {\mathcal{D}}\rangle_t:=\int_{[0,t]\times D^2} Q_u(x,y) \left(1+ \gamma(X_u(x)-\gamma u)\right)\left(1+ \gamma(X_u(y)-\gamma u)\right)
\ind_{\{u\in \mathcal R^x\cap \mathcal R^y\}} \mathcal{G}_u  (\dd x )\mathcal{G}_u(\dd y)\,\dd s.
\end{equation}
 This martingale is bounded in $L^2$ (its bracket is smaller than   that of $(\mathcal{D}_t)_{t\ge 0}$), and thus converges in $L^2$. We call $\tilde{\mathcal{D}}_\infty$ its limit.
 We can compare $\mathcal{D}_\infty$ and $\tilde{\mathcal{D}}_\infty$: as a consequence of \eqref{yop} (when $t\in (T^x_k,R^x_k)$ we have $X_t(x)-\gamma t\geq 0$)
 \begin{equation}
W^0_t(x)\geq W_t(x)- \sum_{i=1}^{I_t^x}\big(A+\eta   T^x_i\big)e^{\gamma A+(\gamma \eta+\frac{\gamma^2}{2})  T^x_i}  
\end{equation}
where $I_t^x=\sup\{i:T^x_i\leq t\}$. Hence if one sets 
$$Q:=\int_D\Big(\sum_{i=1}^{\infty}\mathbf{1}_{\{T_i^x<+\infty\}}\big(A+\eta  T^x_i\big)e^{(\gamma \eta+\frac{\gamma^2}{2})  T^x_i}  \Big)\,   \mu(\dd x),$$
we have
\begin{equation}
\mathcal{D}_\infty\geq \tilde{\mathcal{D}}_\infty-e^{\gamma A}Q.
\end{equation}

To prove Theorem \ref{lefttail},  it is sufficient to show that both terms in the r.h.s. display Gaussian concentration.

\begin{lemma}\label{lem1}
There exists $C(A,\eta,\gamma)$ such that 
$$\forall \alpha\in\R,\quad \E[e^{\alpha \tilde \cD_\infty}]\leq e^{C(A,\eta,\gamma ) \alpha^2 },$$
where 
\begin{equation}
C(A,\eta,\gamma ):= 
\int_{[0,t]\times D^2} 
[\gamma(A+\eta u)+1]^2 e^{2\gamma A+(2\gamma \eta+ \gamma^2)u}Q_u(x,y) \mu(\dd x)\mu( \dd y).
\end{equation}
\end{lemma}

\begin{lemma}\label{lem2}
There exists $C(A,\eta,\gamma)$ such that 
$$\forall \alpha\in\R,\quad \E[e^{\alpha Q}]\leq e^{C(A,\eta,\gamma ) \alpha^2 }.$$
\end{lemma}

 Lemma \ref{lem1} is by far the easier of the two results as the construction of the martingale $\tilde D$ has been tailored so that its 
 quadratic variation is bounded.  Lemma \ref{lem2} requires more work, but the main idea is to control the total variation of the Doob Martingale associated with 
 $Q$ and $\cF_t$ (which is continous), the details are provided in the next session
 
\begin{proof}[Proof of Lemma \ref{lem1}]
From \eqref{expomar} it is sufficient to show that 
$\langle\tilde{\mathcal{D}}\rangle_\infty$ is bounded.
Using the definition of  $\mathcal{R}^{x}$
and the fact (which can be verified by checking that the largest possible negative value is smaller in absolute value that the r.h.s. below) that 
\begin{equation}
 \max_{u\le A+(\gamma+\eta)t} \left| [\gamma(u-\gamma t)+1] e^{\gamma u-\frac{\gamma^2 t}{2}}\right|
 =[\gamma (\eta t+A)+1] e^{\frac{\gamma^2 t}{2}+\gamma[\eta t+A]},
\end{equation}
we obtain that the bracket \eqref{brackettilde} of  $\tilde{\mathcal{D}}_t$ satisfies 
\begin{align*}
\langle\tilde{\mathcal{D}}\rangle_t\leq &\int_{[0,t]\times D^2} 
[\gamma(A+\eta u)+1]^2 e^{2\gamma A+(2\gamma \eta+ \gamma^2)u}Q_u(x,y) \mu(\dd x)\mu(\dd y),
\end{align*}
which is uniformly bounded in $t$ by item 4 of Assumption \ref{ass:WN}, and our choice of $\eta$ \eqref{asumpeta}.
\end{proof}
 
 \medskip

\subsection{Proof of Lemma \ref{lem2}}

To prove concentration for $Q$,  we consider  the Doob martingale associated with  $Q$
(recall $\cF_t:= \sigma( X_s, s\in [0,t])$),
$$Q_t:=\E[Q|\mathcal{F}_t],$$
(we need to prove first that $\E[Q]<+\infty$ but this is the easier part of the proof),
and prove that  $\esssup \langle Q \rangle_{\infty}<\infty$ so that one can conclude 
using the following identity which is provided by the exponential martingale
\begin{equation}
 \bbE\left[e^{\alpha Q- \frac{\alpha^2\langle Q \rangle_{\infty}}{2}}\right]\le \bbE[e^{\alpha \bbE[Q]}].
\end{equation}

We let $Q^x$ denote the integrand in the definition of $Q$ and $Q^x_t$ to be the Doob martingale associated with it,  
\begin{equation}
Q^x:=\sum_{i=1}^{\infty}\mathbf{1}_{\{T_i^x<+\infty\}}\big(A+\eta  T^x_i\big)e^{\gamma A+(\gamma \eta+\frac{\gamma^2}{2})  T^x_i} ,\quad \text{ and }\quad Q^x_t:=\E[Q^x|\mathcal{F}_t].
\end{equation}
Our first task is to bound the expectation of $Q^x$, in a way which is uniform over $x$ so that we have $\bbE[Q]<\infty$.
This computation is also going to be used later to control the martingale bracket.
\begin{equation}\label{lexpec}
Q^x\le \sum_{n=1}^\infty\sum_{i=1}^{\infty}\big(A+\eta   n\big)e^{(\gamma \eta+\frac{\gamma^2}{2})  n} \mathbf{1}_{\{  T^x_i\in (n-1,n]\}}=:\sum_{n=1}^\infty Q^{x,n}.
\end{equation} 
We are going to prove 
\begin{equation}\label{boundexp}
\E\left[\#\{ i \ :  \ T_i\in(n-1,n]\}\right]\leq  \frac{4}{\sqrt{2\pi n}(\gamma+\eta)} e^{-\frac{(\gamma+\eta)^2n}{2}}.
\end{equation}
This implies that  $\E[Q^{x,n}]\leq C \sqrt{n} e^{-\frac{ \eta^2n}{2}}$, and, by linearity,  that $\bbE[Q]$ is bounded.
Let us now prove \eqref{boundexp}. Assuming that $A$ is chosen larger than $ \gamma+\eta$ yields
 \begin{multline}\label{firsttime}
 \P\left(\exists i, \    T^x_i\in (n-1,n]\right)\leq \P(\sup_{s\in[0,n]}X_s(x)\geq (\gamma+\eta)(n-1)+A )\\
 \leq \P\big(\sup_{s\in[0,n]}X_s(x)\geq (\gamma+\eta)n\big)
 \leq   \frac{2}{\sqrt{2\pi n}(\gamma+\eta)} e^{-\frac{(\gamma+\eta)^2n}{2}}
 \end{multline}
 where the last inequality is the standard Gaussian tail estimate. 
Using the Markov property for the Brownian Motion $(X_t(x))_{t\ge 0}$ 
at the $k$-th $T^x_i$ in the interval $(n-1,n]$  the reader can check that
that if $A \ge \gamma +2$
\begin{equation}
 \P[\#\{i \ : \ T_i\in(n-1,n]\} \geq k+1 \ | \ \#\{i \ : \  T_i\in(n-1,n]\}\geq k]\leq \bP[\inf_{t\in [0,1]} B_t\le \gamma-A ]\le \frac{1}{2}.
 \end{equation}
 Therefore we have from \eqref{firsttime}
 \begin{multline}
  \E\left[\#\{ i \ :  \ T_i\in(n-1,n]\}\right]= \sum_{k=1}^{\infty} \P[\#\{i \ : \ T_i\in(n-1,n]\} \geq k]\\
  \le \P[\#\{i \ : \ T_i\in(n-1,n]\} \geq 1]  \sum_{k=1}^{\infty}  2^{1-k}\le \frac{4}{\sqrt{2\pi n}(\gamma+\eta)} e^{-\frac{(\gamma+\eta)^2n}{2}}.
 \end{multline}

 \medskip
 Now we focus on controlling the martingale bracket. As  $(Q^x_t)$ is also martingale with respect to the Brownian Filtration associated with
 $(X_t(x))$, we can (from \cite[Chap. V Th. 3.4]{cf:RY}) write its variation in the form
 \begin{equation}\label{arop}
 \dd Q^x_t:= A^x_t \dd X_t(x).
 \end{equation}
Then using the covariance structure of $X_t(x)$ one can compute the infinitesimal increment of the martingale bracket 
 \begin{equation}\label{ladaic}
 \dd  \langle Q^x, Q^y  \rangle_t=   A^x_t A^y_t Q_t(x,y) \dd t,
 \end{equation}
and thus obtain an expression for  $\langle Q \rangle_{\infty}$
\begin{equation}
 \langle Q \rangle_{\infty}= \int_{\bbR_+\times D^2} A^x_u A^y_u Q_u(x,y) \mu(\dd x)\mu( \dd y)\,\dd u.
\end{equation}

Our remaining task is to find an expression for $A^x_t$ and obtain a uniform bound in $x$ for it. 
We must distinguish between two cases according to whether $t\in (  T_k^x,  R^x_k)$ or $t\in (  R^x_k,  T^x_{k+1})$. 
 
 In the first case, using the strong Markov property for the Brownian motion  $(X_t(x))_{t\ge 0}$ 
 and denoting by $\bE_{z}$ the law of standard Brownian  motion $(B_t)_{t\ge 0}$ starting at $z$ we have
 
 \begin{equation}
Q^x_t=\sum_{i=1}^{k}\big(A+\eta   T^x_i\big)e^{(\gamma \eta+\frac{\gamma^2}{2})  T^x_i}+
\bE_{X_t(x)}\Big[\sum_{i=1}^{\infty}\mathbf{1}_{\{\hat T^t_i<+\infty\}}\big(A+\eta ( \hat T^t_i+t)\big)e^{(\gamma \eta+\frac{\gamma^2}{2})( \hat T^t_i+t)}\Big]
\end{equation}
where the sequence $\hat T^t_k=\hat T^t_k(B)$ (we drop the dependence in $B$ to alleviate the notation) is recursively defined by  $\hat T_0^t:=0$ and for $k\geq 1$
\begin{align*}
 \hat R^t_k:=& \inf\{s\geq \hat T^t_{k-1} \ : \ B_s\leq \gamma(t+s)\},\\ 
 \hat T^t_k:=& \inf\{s\geq \hat R^t_{k} \ : \ B_s=A+(\eta+\gamma)(t+s)\}
\end{align*}
From this we deduce that \eqref{arop} holds for $t\in (  T_k^x,  R^x_k)$ with
\begin{equation}\label{lunz}
 A^x_t:=
 \partial_z \left( \bE_{z}\Big[\sum_{i=1}^{\infty}\mathbf{1}_{\{\hat T^t_i<+\infty\}}\big(A+\eta ( \hat T^t_i+t)\big)e^{(\gamma \eta+\frac{\gamma^2}{2})( \hat T^t_i+t)}\Big] \right) _{\big|z=X_t(x)}.
\end{equation}
In the second case, the same argument yields 
 \begin{equation}
Q^x_t=\sum_{i=1}^{k}\big(A+\eta   T^x_i\big)e^{(\gamma \eta+\frac{\gamma^2}{2})   T^x_i}+\bE_{X_t(x)}\Big[\sum_{i=1}^{\infty}\mathbf{1}_{\{T_i^t<+\infty\}}
\big(A+\eta ( T^t_i+t)\big)
e^{(\gamma \eta+\frac{\gamma^2}{2})( T^t_i+t)}\Big]
\end{equation}
where the sequence $T^t_k$ is recursively defined  by $R_0^t:=0$ and for $k\geq 1$
\begin{align*}
T^t_k:=& \inf\{s\geq  R^t_{k-1} \ : \ B_s=A+(\eta+\gamma)(t+s)\}\\
  R^t_k:=& \inf\{s\geq  T^t_{k} \ : \ B_s= \gamma(t+s)\}.
 \end{align*}
Again, we deduce that \eqref{arop} holds for $t\in [  R^x_k,  T^x_{k+1})$ with
\begin{equation}\label{laut}
 A^x_t:=
 \partial_z  \left( \bE_{z}\Big[\sum_{i=1}^{\infty}\mathbf{1}_{\{T_i^t<+\infty\}}
\big(A+\eta ( T^t_i+t)\big)
e^{(\gamma \eta+\frac{\gamma^2}{2})( T^t_i+t)}\Big] \right)_{\big|z=X_t(x)}.
\end{equation}

To conclude, we need to prove the following bounds 

\begin{lemma}\label{lepti}
 We have for every  $t\ge 0$ and $z\ge \gamma t$
 \begin{equation}\label{leew}
 \left|\partial_z\left( \bE_{z}\Big[\sum_{i=1}^{\infty}\mathbf{1}_{\{\hat T^t_i<+\infty\}}\big(A+\eta ( \hat T_i^t+t)\big)e^{(\gamma \eta+\frac{\gamma^2}{2})( \hat T_i^t+t)}\Big]\right)\right|
 \le C(A,\eta, \gamma)(t+1)e^{(\gamma \eta+\frac{\gamma^2}{2})t}.
 \end{equation}

 \end{lemma}

 \begin{lemma}\label{leghran}
   We have for every  $t\ge 0$ and $z\le (\gamma +\eta t)$ 
   \begin{equation}\label{ippp}
    \left|\partial_z\left( \bE_{z}\Big[\sum_{i=1}^{\infty}\mathbf{1}_{\{T_i^t<+\infty\}}\big(A+\eta (T^t_i+t)\big)e^{(\gamma \eta+\frac{\gamma^2}{2})( T^t_i+t)}\Big]\right)\right|
    \le C(A,\eta, \gamma)(t+1)e^{(\gamma \eta+\frac{\gamma^2}{2})t}.
   \end{equation}
\end{lemma}
These two results together with \eqref{ladaic} and \eqref{lunz}-\eqref{laut} yields the result (using the fact that $2\gamma \eta+\gamma^2<2$ and 
the fourth item in Assumption \ref{ass:WN}).\qed

\begin{proof}[Proof of Lemma \ref{lepti}]
 
We need to show that the functional $\bE_z[\cdot]$ in the r.h.s. of \eqref{leew} is Lipschitz in $z$ for an adequate Lipschitz constant.
We fix thus $x,y \ge \gamma t$ (we may consider only the case where $|x-y|<\gep$ for an arbitrary $\gep>0$)
and consider $\bbP$ to be the coupling between $\bE_x$ and $\bE_y$ constructed with marginals that evolve independently until the first time they meet
$$\tau:= \inf\{ s>0 \ : \ B^{(1)}_s= B^{(2)}_s\},$$
and jointly afterwards.
We let $\hat T^{(t,j)}_i$, $\hat R^{(t,j)}_i$,  $j=1,2$ denote the stopping time corresponding to each coordinate of the coupling (i.e. \ 
$\hat T^{(t,j)}_i:= \hat T^{t}_i(B^{(j)})$).
We have (again, here $\bE_z[\cdot]$ corresponds to the expression in \eqref{leew})
\begin{multline}\label{lffs}
 \left|\bE_{x}[\cdot]-\bE_{y}[\cdot]\right|  
\le  \bbP \left[ \tau>\min(\hat T^{(t,1)}_1,\hat T^{(t,2)}_1,1)\right]\\ \times
\bbE\left[ \sum_{j=1,2}
 \sum_{i=1}^{\infty}\mathbf{1}_{\{  T_i^{(t,j)}<+\infty\}}\left[ A+\eta (\hat T^{(t,j)}_i+t)\right]e^{\left(\gamma \eta+\frac{\gamma^2}{2}\right)(\hat T^{(t,j)}_i+t)} \ | \ \tau>\min(\hat T^{(t,1)}_1,\hat T^{(t,2)}_1,1) \right].
 \end{multline}
The first step is to bound the probability above. We are going to show that for some constant $C$ (which does not depend on $t$) we have
 \begin{equation}
\bbP \left[ \tau>\min(\hat T^{(t,1)}_1,\hat T^{(t,2)}_1,1)\right]\leq C|x-y|.
\end{equation}
 The bound  $\bbE \left[ \tau> 1\right]< (C/3)|x-y|$ is standard and thus
 by symmetry and union bound we only need to show that 
\begin{equation}\label{ustimate}
\bbE \left[ \tau> \hat T^{(t,1)}_1\right]\le \frac{C|x-y|}{3}.
\end{equation}
Note that provided $A\ge 2$, for any choice of $x\ge \gamma t$ we have 
\begin{equation}\label{stuf}
\hat T^{(t,1)}_1\ge \min\{ s \ : \  |B^{(1)}_s-x|\ge 1\},
\end{equation}
and hence 
\begin{equation}
 \bbE \left[ \tau> \hat T^{(t,1)}_1\right] \le \bP_{(0,y-x)} \left[ \cT_{\gD}> \cT_{\{-1,1\}\times \bbR}   \right]\le C|x-y|.
\end{equation}
where $\gD:=\{ (x,x) : x\in \bbR\}$ and $\cT_A$ denotes the hitting time of a set $A$ by a two dimensional Brownian motion (here with initial condition $(0,y-x)$).
The last inequality can be deduced by standard estimate for Brownian Motion (see Appendix \ref{hitting}).

\medskip

To complete the proof we need to show that the conditional expectation in \eqref{lffs} satisfies
\begin{equation}
\bbE\left[ \sum_{j=1,2}
 \sum_{i=1}^{\infty} \dots  \ | \ \tau>\min(\hat T^{(t,1)}_1,\hat T^{(t,2)}_1,1) \right] \le C (t+1)e^{(\gamma \eta+\frac{\gamma^2}{2})t}.
 \end{equation}
To do so we use the Markov property for $(B^{(1)}_t,B^{(2)}_t)_{t\ge 0}$ at time $\cT=\min(\hat T^{(t,1)}_1,\hat T^{(t,2)}_1,1)$
and consider the supremum over all possible realizations $(r,z)$ of $(\cT, B^{(j)}_{\cT})$ and obtain

\begin{multline}
\bbE\left[
 \sum_{i=1}^{\infty}\mathbf{1}_{\{ T_i^{(t,j)}<+\infty\}}\left[ A+\eta (\hat T^{(t,j)}_i+t)\right]e^{\left(\gamma \eta+\frac{\gamma^2}{2}\right)(\hat T^{(t,j)}_i+t)} \ | \ 
 \tau>\min(\hat T^{(t,1)}_1,\hat T^{(t,2)}_1,1) \right]\\
 \le \sup_{r\in [t,t+1]}
 \max 
 \Bigg( \sup_{z> \gamma r}
 \bE_z\left[ \sum_{i=1}^{\infty}\mathbf{1}_{\{\hat T_i^r<+\infty\}}\left[ A+\eta (\hat T^r_i+r)\right]e^{\left(\gamma \eta+\frac{\gamma^2}{2}\right)(\hat T^{r}_i+r)}\right],\\
 \sup_{z\le A+ (\gamma+\eta) r}
 \bE_z\left[ \sum_{i=1}^{\infty}\mathbf{1}_{\{ T_i^r<+\infty\}}\left[ A+\eta(T^{r}_i+r)\right]e^{\left(\gamma \eta+\frac{\gamma^2}{2}\right)(T^{r}_i+r)}\right] \Bigg)=\cL(t).
\end{multline}
The two terms in the $\max$  lead us to make the distinction between the two cases $\hat R^{(t,j)}_1<\cT$ and  $\hat R^{(t,j)}_1>\cT$. 
Using the Markov property at time $R^{(t,j)}_1$ in the first case, we obtain that for any $z> \gamma r$ we have
\begin{multline}
 \bE_z\left[ \sum_{i=1}^{\infty}\mathbf{1}_{\{ T_i^r<+\infty\}}\left[ A+\eta (\hat T^r_i+r)\right]e^{\left(\gamma \eta+\frac{\gamma^2}{2}\right)(\hat T^{r}_i+r)}\right]
 \\ \le \sup_{s\ge r} \bE_{\gamma s}\left[ \sum_{i=1}^{\infty}\mathbf{1}_{\{ T_i^r<+\infty\}}
 \left[ A+\eta (T^s_i+s)\right]e^{\left(\gamma \eta+\frac{\gamma^2}{2}\right)(T^{s}_i+s)}\right].
\end{multline}
Now repeating the computation \eqref{lexpec}-\eqref{firsttime} we have 

\begin{multline}\label{ileborn}
  \bE_{\gamma s}\left[ \sum_{i=1}^{\infty}\mathbf{1}_{\{ T_i^s<+\infty\}}\left[ A+\eta (T^s_i+s)\right]e^{\left(\gamma \eta+\frac{\gamma^2}{2}\right)(T^{s}_i+s)}\right]
  \le C \sum_{n\ge 1}  (n+s) e^{\left(\gamma \eta+\frac{\gamma^2}{2}\right)(n+s)} \bE_{\gamma s}\left[\# \{ T^s_i\in [n-1,n]\}\right]\\
  \le 2C \sum_{n\ge 1}  (n+s) e^{\left(\gamma \eta+\frac{\gamma^2}{2}\right)(n+s)}\bP_{\gamma s}\big[\exists i, T^s_i\in [n-1,n]\big]
 \\ \le 2C \sum_{n\ge 1}  \frac{(n+s)}{\sqrt{n}} e^{\left(\gamma \eta+\frac{\gamma^2}{2}\right)(n+s)-\frac{ (\eta s+(\gamma+\eta)n)^2 }{2 n}}.
\end{multline}
Using Laplace's method the last quantity can be bounded above by $e^{\left(\frac{\gamma^2}{2}-2\eta^2\right)s}$, which is uniformly bounded in $s$
(from \eqref{asumpeta} we have $\frac{\gamma^2}{2}-2\eta^2<0$).

\medskip

In the second case, repeating again \eqref{lexpec}-\eqref{firsttime} we have
\begin{multline}\label{dzeko}
 \bE_z\left[ \sum_{i=1}^{\infty}\mathbf{1}_{\{ T_i^r<+\infty\}}\left[ A+\eta(T^{r}_i+r)\right]e^{\left(\gamma \eta+\frac{\gamma^2}{2}\right)(T^{r}_i+r)}\right]
 \\ \le C\sum_{n\ge 1} (n+r) e^{\left(\gamma \eta+\frac{\gamma^2}{2}\right)(n+r)} \bP_{z}\left[ \sup_{t\in [n-1,n]} B_t\ge (\gamma+\eta)(n-1+r)+A\right]
 \le C'  r e^{\left(\gamma \eta+\frac{\gamma^2}{2}\right)r},
 \end{multline}
 which yields  
\begin{equation}
 \cL(t)\le C(A,\eta,\gamma) (t+1)  e^{\left(\gamma \eta+\frac{\gamma^2}{2}\right)t}.
\end{equation}

\end{proof}

\begin{proof}[Proof of Lemma \ref{leghran}]

Like in the previous lemma, we prove \eqref{ippp} by providing Lipschitz bounds  in the same way as \eqref{lffs}. It is of course sufficient to consider the case
when $|x-y|\le 1$.
The approach adopted in the proof of Lemma \ref{lepti} does not work when both $x$ and $y$ are close to $A+(\gamma+\eta)t$, because in that case, with an
independent coupling
the probability of the two motions merging before $\min( T ^{(t,1)}_1,T^{(t,2)}_1,1)$ can be made arbitrarily small.

\medskip

Note however that this obstruction is not present if one decides to start the sum from $i=2$. 
Setting \begin{equation}\label{defg}
         g_t(u):= \big(A+\eta (u+t)\big)e^{(\gamma \eta+\frac{\gamma^2}{2})( u+t)}
        \end{equation}

and repeating the proof Lemma \ref{lepti} one obtains

\begin{equation}
 \left|\bE_{x}\left[\sum_{i=2}^{\infty}\mathbf{1}_{\{ T_i^t<+\infty\}}g_t(T_i^t)\Big]\right]-
 \bE_{y}\left[\sum_{i=2}^{\infty}\mathbf{1}_{\{ T_i^t<+\infty\}}g_t(T_i^t)\Big]\right]\right|
  \le 
 C(A,\eta,\gamma)|x-y| (t+1)  e^{\left(\gamma \eta+\frac{\gamma^2}{2}\right)t}.
\end{equation}

Then we must estimate the difference between the first terms and show that 
\begin{equation}\label{dafirst}
  \Bigg|\bE_{x}\left[ \mathbf{1}_{\{ T_1^t<+\infty\}}g_t(T^{t}_1)\right]
  -
  \bE_{y}\left[\mathbf{1}_{\{ T_1^t<+\infty\}}g_t(T^{t}_1)\right]\Bigg|
  \le  C(A,\eta,\gamma)|x-y| (t+1)  e^{\left(\gamma \eta+\frac{\gamma^2}{2}\right)t}.
\end{equation}
To prove \eqref{dafirst} we choose to couple the two Brownian motions in a ``parallel'' fashion 
$$B^{(2)}_s=B^{(1)}_s+(y-x).$$
To alleviate the notation  we let $S_1$ and $S_2$ designate the hitting time $T^{t}_1(B^{(1)})$ and $T^{t}_1(B^{(2)})$.
Let us assume without loss of generality that $y<x$. 
Note that as $g_t(\cdot)$ is an increasing function and $S_2>S_1$ we have
\begin{equation}
\bbE\left[ \mathbf{1}_{\{S_2<\infty\}} g_t(S_2)\right]\ge 
\bbE\left[ \mathbf{1}_{\{S_1<\infty\}}g_t(S_1)\right]
\bbP\left[ S_2<\infty \ | \ S_1<\infty \right].
\end{equation}
and thus the difference in \eqref{dafirst} satisfies
\begin{equation}\label{lipz}
 \bE_x[ \mathbf{1}_{\{ T_1^t<+\infty\}}g_t(T^{t}_1)  \cdot]-\bE_y[ \mathbf{1}_{\{ T_1^t<+\infty\}}g_t(T^{t}_1)  \cdot]\le
 \bbP\left[ S_2=\infty \ | \ S_1<\infty \right]\bbE\left[ \mathbf{1}_{\{S_1<\infty\}}g_t(S_1) \right]
\end{equation}
And we have already proved (cf. \eqref{dzeko}) that
\begin{equation}\label{loopz}
\bbE \left[\mathbf{1}_{\{S_1<\infty\}}g_t(S_1)\right]\le  
C(A,\eta,\gamma) (t+1)  e^{\left(\gamma \eta+\frac{\gamma^2}{2}\right)t},
\end{equation}
Now using the strong Markov property for $B^{(2)}$ at $S_1$ we obtain that 
\begin{equation}
\bbP\left[ S_2<\infty \ | \ S_1<\infty \right]=\bP_{y-x}[\exists s, B_s=(\gamma+\eta)s]= e^{-2(\gamma+\eta)(x-y)},
\end{equation}
where the last inequality is obtained by observing that $u(x):=\bP_{x}[\cdot]$ is a solution of $$u''(x)-2(\gamma+\eta)u'(x)=0.$$
Combining this with \eqref{lipz} and \eqref{loopz} we obtain that 
\begin{equation}
 \bE_x[ \mathbf{1}_{\{ T_1^t<+\infty\}}g_t(T^{t}_1)]-\bE_y[ \mathbf{1}_{\{ T_1^t<+\infty\}}g_t(T^{t}_1) ]
 \le C |x-y| (t+1)  e^{\left(\gamma \eta+\frac{\gamma^2}{2}\right)t}.
 \end{equation}

\medskip

For the other bound, let us define $\cT:= \min\{S_2, S'_2\}$
where 
$$S'_2:= \inf\{ s> S_1 \ : \ B^{(2)}_s= \gamma(t+s)\}.$$
To estimate $\bE_y[ \mathbf{1}_{\{ T_1^t<+\infty\}}g_t(T^{t}_1) ]$,
we split it into two contribution, depending how $S_2$ compares to $S'_2$.
We show that 
\begin{equation}\begin{split}\label{toosho}
\bbE\left[ g_t(S_2) \ind_{\{S_2<S'_2<\infty\}}\right]&\le \bbE\left[ \ind_{\{S_1<\infty\}} g_t(S_1)\right],\\
\bbE\left[ g_t(S_2) \ind_{\{S'_2<S_2<\infty\}}\right]&\le 
 C|x-y|.
\end{split}\end{equation}
and the sum of these two inequalities yields (note that $S'_2<\infty$ with probability $1$)
\begin{equation}
 \bE_y[ \mathbf{1}_{\{ T_1^t<+\infty\}}g_t(T^{t}_1)]-\bE_x[ \mathbf{1}_{\{ T_1^t<+\infty\}}g_t(T^{t}_1) ]
 \le C' |x-y|.
 \end{equation}

Let us consider the Brownian Motion $\tilde B_s:= B^{(2)}_{S_1+s}-B^{(2)}_{S_1}$.
Note that $\cT-S_1$  is a stopping time for $\tilde B$.

Note that, conditioned on the event $S_1<\infty$, $K_s:=\left[B^{(2)}_{S_1+s}-\gamma(t+s)\right]e^{\gamma (B^{(2)}_{s+S_1}-A)-\frac{\gamma^2 (s+t)}{2}}$ 
is a martingale for the filtration $\tilde \cF$ defined by $$\tilde \cF_u:= \cF_{S_1}\cup\sigma( \tilde B_s, s\le u),$$ and thus
$K_{s\wedge(\cT-S_1)}$ is a positive martingale. 
Using Fatou's Lemma for the conditional expectation with respect to  (for the first inequality)
\begin{multline}
 \bbE \left[\ind_{\{S_2<S'_2<\infty\}}\left[A+\eta (S_2+t)\right]e^{\left(\gamma \eta+\frac{\gamma^2}{2}\right)(S_2+t)} \right]
 =  \bbE\left[  \bbE \left[K_{\cT-S_1}  \ | \cF_{S_1}  \right] \ind_{\{S_1<\infty\}} \right]
 \\
\le  \bbE \left[ \ind_{\{S_1<\infty\}}\left[B^{(2)}_{S_1}-\gamma(t+s)\right]e^{\gamma (B^{(2)}_{S_1}-A)-\frac{\gamma^2 (s+t)}{2}} \right]
\le \bbE \left[\ind_{\{S_1<\infty\}} g_t(S_1)\right].
\end{multline}
The second inequality is a consequence of the fact that $B^{(2)}_{S_1}=(\gamma+\eta)(S_1+s)+A-(x-y)$.  
To prove the second inequality in \eqref{toosho}, we use the Markov property at time $S'_2$ which yields
\begin{equation}
    \bbE \left[g_t(S_2) \ind_{\{S_2<\infty\}} \ | \  S'_2\ \right]\ind_{\{S'_2<S_2\ ;\ S'_2<\infty\}}
    =\tilde \bE_{B_{S'_2}}\left[ g_{t+S'_2}(T^{t+S'_2}_1(\tilde B))\ind_{\{T^{t+S'_2}_1<\infty\}}\right].
\end{equation}
where $\tilde B$ denote a Brownian Motion independent of $B$.
Setting $t+S'_2=s$ and noting that $B_{S'_2}=\gamma s$, we obtain
\begin{equation}
   \bbE \left[g_t(S_2) \ind_{\{S'_2<S_2<\infty\}}  \right] \leq
 \bbP[   S'_2<\infty ] \max_{s\ge t}\bE_{\gamma s}\left[ g_{s}(T^s_1)\ind_{\{T^s_1<\infty\}}\right].
\end{equation}
The second factor is uniformly bounded (cf. \eqref{ileborn}).
To control the first one we can control the conditional expectation since 
\begin{multline}
 \bbP[ S'_2<S_2 \ | \  S_1<\infty ]= \bP_{y-x}\left[ \min\{ s \ : \ B_s=
 \gamma s-A-\eta t\} < \min\{ s \ : \ B_s=
 (\gamma+\eta)s\} \right]\\
 \le \bP_{y-x}\left[ \min\{ s \ : \ B_s=
 \gamma s-A-\eta t\} < \min\{ s \ : \ B_s=
 \gamma s\} \right]= u(y-x)
\end{multline}
where $u$ is the solution of the equation
$$u''(x)-2 \gamma u'(x)=0,$$
with boundary condition $u(0)=0$, $u(-A-\eta t)=1$.
A simple computation yields
 \begin{equation}
 u(y-x)=\frac{1-e^{-2 \gamma (x-y)}}{1-e^{-2 \gamma(A+\eta t)}}\le C|x-y|
 \end{equation}
and completes our proof.

\end{proof}

\section{Small deviations of the GMC measure: Proof of Theorem \ref{laprop}}\label{sec:small}

 \subsection{Sketch of proof}

 We assume the setup described in Section \ref{bdec} and Assumption \ref{ass:WN} (item 4 is not required).

 We set $\tilde X=X-Z$ (recall that $Z=\frac{1}{\mu(D)}\int_D X(x) \dd x)$ 
  let $\tilde X_t$ denote the corresponding martingale approximation,
 \begin{equation}
 \tilde{X}_t(z)=X_t(z) - \int_D X_t(x) \bar h(x) \dd x,
 \end{equation}
 where $\bar h= h/ \mu(D)$.
In view  of \eqref{inzeform} the covariance structure of $\tilde X$ is given by
\begin{equation}
 \bbE\left[ \tilde X_t(x)\tilde X_s(y) \right]=\int^{s\wedge t}_0 \tilde Q_u (x,y) \dd u.
\end{equation}
where $\tilde Q$ is defined in \eqref{deftildeQ} and satisfies the same assumptions as $Q$. 

\medskip

Now using the fact that $\cG_{\infty}$ as the limit of $\cG_t$ \eqref{defgt}
we can express  $e^{-\gamma Z}\cG_{\infty}$ as the following martingale limit
\begin{equation}
 e^{-\gamma Z}\cG_{\infty}= \lim_{t\to \infty} \int_D e^{\gamma \tilde X_t(x) - \frac{\gamma^2}{2}\bbE[\tilde X_t(x)]  } e^{\frac{\gamma^2}{2}q(x)}h(x) \dd x.
\end{equation}
It is obviously sufficient to prove the result for the limit above with $e^{\frac{\gamma^2}{2}q(x)}h(x)$ replaced by $1$ (since the function is bounded from below) and for the sake of readability we assume that $D=[0,1]^2$ and $\bar h=1$.
We set in that case 
$\tilde M_t:= \int_D e^{\gamma \tilde X_t(x) - \frac{\gamma^2}{2}\bbE[\tilde X_t(x)]  } \dd x.$
and Theorem \ref{laprop} reduces to showing 
\begin{equation}\label{amontrer}
 \P \left (   \tilde M_{\infty}\le  s  \right )  \leq 2 e^{-c |\log s|^{-\beta} s^{-\frac{4}{\gamma^2}}}
\end{equation}

Before stating the main idea of the proof, let us make a   trivial observation: 
assuming that  $D$ is the square $[0,1]^2$, we have for any  smooth Gaussian field $\tilde X$ 
(as opposed to the $\log$-correlated field which is only defined as a distribution), which satisfies  $\int _D\tilde X(x)\dd x=0$,  
 $$  \int_D e^{  \gamma \tilde{X}} \dd x\ge e^{ \int_D  \gamma \tilde{X}\,\dd x }=1. $$
 Things do not quite work as easily   when considering the exponential of our $\log$-correlated field  because in the exponential 
we are subtracting the variance which is infinite. However we can still control the probability of being small in two steps. Heuristically this goes as follows: 
 \begin{itemize}
 \item [(A)] We apply Jensen to the field $\tilde X_t$.
to show that $\tilde M_t \ge e^{-\frac{\gamma^2t}{2}}$, for a value of $t$ such that $e^{-\frac{\gamma^2t}{2}}\ge 2s$.
 \item [(B)] We observe that the field $\tilde X -\tilde X_t$ which remains to be added to obtain $\tilde M$, 
 has a small covariance when the distance is larger than $e^{-t}$ so that conditioned on $\tilde M_t$, the r.v.
  $$\tilde M_{\infty} = \int e^{  \gamma (\tilde{X}-\tilde{X}_t)-\frac{\gamma^2}{2}\E[(\tilde{X}-\tilde{X}_t)^2]}  \dd \tilde M_t$$
 is the sum of a large number (of order $e^{2t}$) of almost independent positive contribution coming from regions of diameter $e^{-t}$, 
 and thus should concentrate around its mean, which according to step $(A)$ is larger than $2s$.
 \end{itemize}
 A quantitative implementation of this heuristic yields the desired exponent $4/\gamma^2$.
 
 \medskip
 
 Turning this idea into a rigorous proof requires some care for the following reason: 
 the independent variables appearing in the second step have inhomogeneous weight (this is the effect of $\tilde X_t$).
Instead of using Jensen in step $(A)$ we rely on the following observation (see Lemma \ref{leem}): either $\tilde X_t$ is larger or equal to $-1$ on most of $D$, 
or there exists a small region on which $\tilde X_t$ takes high value.
In the first case $\tilde X_t$ can be replaced by $-1$ and we do not have to worry about inhomogeneities, 
and concentration in step (B) above can be proved with standard tools (see Proposition \ref{lasuperbelleproposition} below). 
To analyze the case where  $\tilde X_t$ takes high value in a small region, 
we show that the loss of  concentration implied by the smallness of the region is more than compensated by the fact that $\tilde M_t$ is very large, this is the more delicate part of the analysis.

 \subsection{Proof of Theorem \ref{laprop}}

We assume throughout the proof that $s$ is sufficiently small: this restriction only affects the value of the constant $c$  in \eqref{lapropas}.
We fix a parameter $\kappa>1$ and given $s$ one defines $t_0$ by the relation 
\begin{equation}
 e^{-\frac{\gamma^2}{2}t_0}= s |\log s|^{2\kappa}.
\end{equation}
Our first task is to show that $\tilde X_{t_0}$ must either be larger than $-1$ on a large set, 
or assume a very large value on a region of small but still significant size.
This is a simple consequence of the fact that $\int_D \tilde X_{t_0}(x)\,\dd x =0$, but we register it as a lemma, with a convenient formulation to be used in our proof.

\begin{lemma}\label{leem}
 For all $s$ sufficiently large, one of the following must hold
  \begin{itemize}
  \item [(i)] $| \{ x\in D \ : \  \tilde{X}_{t_0}(x) \ge -1    \}|\ge |\log s|^{-\kappa}.$
  \item [(ii)] $\exists n \ge n_0(s):= |\log s|^{\kappa}/10, \   | \{ x\in D \ : \  \tilde{X}_{t_0}(x) \ge  n \} | \ge \frac{1}{ n (\log n)^2}$.
 \end{itemize}

\end{lemma}

\begin{proof}
Let us suppose that $(i)$ does not hold. Set $\mathcal{A}:= \{ x\in D \ : \  \tilde{X}_{t_0}(x) \ge -1   \}$. Since $\tilde{X}_{t_0}$ has zero spatial average, one has
\begin{align*}
\int_D (\tilde{X}_{t_0}(x))_+ dx= & \int_D (\tilde{X}_{t_0}(x))_{-} dx \\
\geq &  \int_{D\cap \mathcal{A}} (\tilde{X}_{t_0}(x))_{-} dx\ge 1-   |\log s|^{-\kappa}.
\end{align*} 
One also gets the bound (by using that for all $u \geq -1$, $ | \{ x \, : \, \tilde{X}_{t_0}(x) \ge  u \}|    \le |\mathcal{A}|\le (\log s)^{-\kappa}$  )
\begin{multline}
\int_D (\tilde{X}_{t_0}(x))_+ \dd x  = \int_0^\infty  | \{ x \ : \  \tilde{X}_{t_0}(x) \ge  u \}|   \dd u  \\
\le \frac{1}{10}  + \int_{(\log s)^{\kappa}/10}^\infty  | \{ x \, : \,  \tilde{X}_{t_0}(x) \ge  u \}|   \dd u  
 \le \frac{1}{10}  + \sum_{n= (\log s)^{\kappa}/10}^\infty  | \{ x \, : \,  \tilde{X}_{t_0}(x) \ge  n \}|   .
\end{multline}
If for all $n \ge  |\log s|^{\kappa}/10$, one has $    | \{ x \, : \,  \tilde{X}_{t_0}(x) \ge  n \}|   \le \frac{1}{ n (\log n)^2}$ then combining the above considerations leads to
\begin{equation*}
1-  |\log s|^{-\kappa} \le  \frac{1}{10}  + \sum_{n= |\log s|^{\kappa}/10}^\infty \frac{1}{ n (\log n)^2} \le  \frac{1}{10}  + \frac{1}{\kappa  \log |\log s | },
\end{equation*}
which is a contradiction. Therefore $(ii)$ holds.

\end{proof}

Now define the events
\begin{equation}\begin{split}
 A&:=\left\{  \, | \{ x\in D  \, : \,  \tilde{X}_{t_0}(x) \ge -1)   \}|\ge (\log s)^{-\kappa}   \, \right\},\\
 B_n&:=\left\{  \, | \{ x\in D \, : \,  \tilde{X}_{t_0}(x) \ge n)   \}|\ge n^{-1}(\log n)^{-2}   \, \right\},\\
 \bar B_n&:= B_n\setminus B_{n-1}, \ \text{ for } n\ge n_0+1 \quad \text{and} \quad  \bar B_{n_0}=B_{n_0}\setminus A.
\end{split}\end{equation}
According to Lemma \ref{leem} the events $(\bar B_n)_{n\ge n_0}$ and $A$ partition the space.
As a consequence we have for any choice of $n_1>n_0$ (in the remainder of the proof we write $\tilde M$ for $\tilde M_{\infty}$)
\begin{multline}\label{loutrouc}
 \P[ \tilde M\le s ] = \P[ \{\tilde M\le s\} \cap A]+ \sum_{n\ge n_0} \P[ \{ \tilde M\le s \} \cap  \bar B_n ]\\
\le \max \left( \bP[\tilde  M\le s \, | \, A], \max_{n\in \lint  n_0,n_1\rint}  \bP[ \tilde M\le s \, | \, \bar B_n] \right)
+  \sum_{n>n_1} \bP[ \bar B_n].
 \end{multline}
Thus we need to control all the conditional probabilities above, with the possibility of discarding high values of $n$ for which $B_n$ is very unlikely.
The adequate choice for $n_1$ turns out to be  $n_1(s):=s^{-\frac {2}{\gamma^2}}$.
Indeed there exists some constant $C>0$ such that for all $x \in D$
\begin{equation*}
\E[  \tilde{X}_{t_0}(x)^2] \le {t_0}+C,
\end{equation*} 
 
hence 
   \begin{equation}
    \bbE \left[   | \{ x\in D \ : \  \tilde{X}_{t_0}(x)  \ge  n  \} |\right] \le  e^{-\frac{n^2}{ 2({t_0}+C)}}.
   \end{equation}
We get from the Markov inequality that 
\begin{equation}
 \bbP[B_n]\le C n (\log n)^2 e^{- c \frac{n^2}{  {t_0}}},
\end{equation}
\begin{equation}
 \sum_{n\ge n_1} \bP\left[ \bar B_n\right]\le e^{- c s^{-\frac 4{\gamma^2}}|\log s|^{-1}}.
 \end{equation}
 
Now to conclude we need to control the first term in the r.h.s. of \eqref{loutrouc}, which amounts to control every conditional expectation.
More precisely, we prove a bound for the conditional expectation with respect to $\tilde X_{t_0}$ which is valid on the specified events.   First we show that almost surely 
\begin{equation}\label{lafirst}
 \bP[ M\le s \, | \,  \tilde X_{t_0}]\ind_A\le e^{ -|\log s|^{-\gb} s^{-\frac{4}{\gamma^4}}}.
\end{equation}
and then that 
\begin{equation}\label{lasecond}
 \bP[ M\le s \, | \,  \tilde X_{t_0}]\ind_{B_n}\le e^{ -|\log s|^{-\gb}s^{-\frac{4}{\gamma^4}}}.
\end{equation}

To prove \eqref{lafirst} we consider $\cA:=\{ x  \ : \ \tilde{X}_{t_0}(x) \ge -1\}$ 
and recall that $|\cA| \ge|\log s|^{-\kappa}$ when $A$ holds.
Hence replacing $\tilde X_{t_0}$ by its lower bound on $\cA$ we have
\begin{equation}
 \tilde M\ge c(\gamma) e^{-\frac{\gamma^2 t_0}{2}} \tilde M^{({t_0})}(\cA).
\end{equation}
where $\tilde M^{({t_0})}$ is chaos measure associated with $\tilde X-\tilde X_{t_0}$ meaning 
\begin{equation}
 \tilde M^{(t_0)}(\dd x)=\lim_{t\to \infty} e^{\gamma(\tilde X_{t}-\tilde X_{t_0})-\frac{\gamma^2}{2}\bbE[ (\tilde X_{t}-\tilde X_{t_0})^2]} \dd x,
\end{equation}
and
$c(\gamma)$ is to absorb various constant  including the fact that the variance of $\tilde X_{t_0}$ is not exactly ${t_0}$ cf. \eqref{lavar}. 
Then \eqref{lafirst} follows from the following general concentration result (proved in the Appendix \ref{siu}).

\begin{proposition}\label{lasuperbelleproposition}

We have for every $t_0\ge 0$, for every $\cA\subset D$, with $|\cA|\ge  e^{-2t_0}$
\begin{equation}
\bbP\left[ \tilde M^{(t_0)}(\cA)\le \frac{|\cA|}{2} \right]\le 2 e^{- c e^{2t_0} t_0^{-2\alpha}|\cA|},
\end{equation}
 \end{proposition}

Now, we must prove \eqref{lasecond}. In the case where $B_n$ holds one sets 
\begin{equation*}
\cB_n:=\{ x  \ : \ \tilde{X}_{t_0}(x) \ge n \}
\end{equation*}
and we have $|\cB_n| \ge \frac{1}{n (\log n)^2}$.
Again, we have the obvious lower estimate
\begin{equation}
 \tilde M\ge c(\gamma )e^{\gamma n-\frac{\gamma^2 t_0}{2}} \tilde M^{(t_0)}(\cB_n).
\end{equation}
Now let us chose $t_1=t_0+ \eta n$ for a small but fixed value of $\eta$ and let
$\bar X$ denote the increment of $\tilde X$ between $t_0$ and $t_1$
\begin{equation*}
\begin{split}
 Y&:=\frac{1}{|\cB_n|} \int_{\cB_n} (\tilde{X}_{t_1}(x)-\tilde{X}_{t_0}(x)) dx\\
\bar X_{t_0,t_1}&:= \bar X:=  (\tilde{X}_{t_1}- \tilde X_{t_0})- Y.
\end{split}
\end{equation*}
Now rewriting $\tilde M^{(t_0)}(\cB_n)$ using $\bar X$ and $Y$ we have 
\begin{equation}
\tilde M\ge c(\gamma)
 e^{\gamma(n+Y) -\frac{\gamma^2 t_1}{2}} \int_{\cB_N}  e^{\gamma \bar X(x)}  \tilde M^{(t_1)}(\dd x).
\end{equation}
To conclude we need to show that the pre-factor is large and that the content of the exponential is in some sense concentrated.
An explicit computation using \eqref{intfunction} yields .
\begin{equation*}
\E[Y^2]=\frac{1}{|\mathcal{B}_n|^2}\int^{t_1}_{t_0} \int_{\cB_n\times \cB_n} \tilde Q_u(z,z') \dd z \dd z'     \le  C n  (\ln n)^2 e^{-2t}
\end{equation*}
Hence for fixed $\delta>0$
\begin{equation}
\bbP[Y\le  -\delta n] \le ce^{- cn (\log n)^{-2} e^{2t}},
\end{equation}
from which we deduce
\begin{equation}
 \bbP\left[ 
 c e^{\gamma n+ \gamma Y -\frac{\gamma^2 t_1}{2}}\le  e^{\gamma n/2} \right]\le ce^{- cn (\log n)^{-2} e^{2t}}\le e^{- s^{-\frac{4}{\gamma^4}}}.
\end{equation}
We are left with showing that 
$$\bbP\left[ \int_{\cB_N}  e^{\gamma \bar X}  \tilde M^{(t_1)}(\dd x) \le e^{-\gamma n/2} \right] $$ is small.
Repeating the reasoning of Lemma \ref{leem} we have either 
\begin{itemize}
 \item [(i)] $|\{ x\in \cB_n \ : \ \bar X\ge -1 \}|\ge  e^{-\sqrt{n}}$,
 \item [(ii)] There exists $m\ge m_0(n):= e^{\sqrt{n}/2}$ such that $|\{ x\in \cB_n \ : \ \bar X\ge m \}|\ge m^{-2}$. 
\end{itemize}
In case $(i)$ setting $\cA_n:= \{ x\in \cB_n \ : \ \bar X\ge -1 \}$ 
we notice that conditioned to $\tilde X_{t_1}$
\begin{equation}
\int_{\cB_n}  e^{\gamma \bar X}  \tilde M^{(t_1)}(\dd x)
\ge e^{-\gamma} \tilde M^{(t_1)}(\cA_n).
\end{equation}
We can then use Proposition \ref{lasuperbelleproposition}
to show that
\begin{equation}
\bbP\left[ \tilde M^{(t_1)}(\cA_n)\le \frac{|\cA_n|}{2} \right] \le e^{-e^{cn}},
\end{equation}
which is more than sufficient.
Finally we need to show that the probability of being in case $(ii)$ above is small.
Note that $\bar X$ has a variance of order $n$
thus we have  
\begin{equation}
 \bbE\left[ |\{ x\in \cB_n \ : \ \bar X\ge m \}| \right] \le e^{-c m^2/n}.
\end{equation}
And hence using Markov inequality (and eventually changing the value of $c$)
\begin{equation}
 \bbP\left[ \exists m\ge m_0(n),\  |\{ x\in \cB_n \ : \ \bar X\ge m \}|\ge m^{-2} \right] \le \sum_{m\ge m_0} m^2 e^{-c m^2/n}\le  c  e^{-c e^{\sqrt{n}}}.
 \end{equation}
Here as we have $n\ge n_0=(-\log s)^{\kappa}$, the right-hand side is smaller than $e^{- s^{-\frac{4}{\gamma^4}}}$ provided $\kappa\ge 2$.

\appendix
\section{Dirichlet GFF and Markov Property}
\label{diriche}
\subsection*{Dirichlet GFF}
We refer to   \cite[Section 4.2]{dubedat}   for references concerning this subsection. Consider a bounded simply connected domain $D\subset\C$ equipped with a smooth (up to the boundary of $D$) conformal metric $g:=e^{\omega(z)}|dz|^2$. Define the Sobolev space $H^1_0(D)$ as the closure  of smooth functions with compact support in $D$, i.e. $C^\infty_0(D)$, with respect to the Dirichlet energy $\int |df|^{\red 2}_g\,{\rm dv}_g$. This space does not depend on the particular choice of $g$   by conformal invariance of the Dirichlet energy. The Dirichlet Green function $G$ is the integral kernel for the mapping $f\in L^2(D,{\rm dv}_g) \mapsto u\in H^1_0(D)$ defined by 
$$-\Delta_g u=  f,\quad u\in H^1_0(D),$$
that is
$$u(x)=\int_DG(x,y)f(y) {\rm v}_g(dy).$$
Again the Dirichlet Green function does not depend on the choice of the metric $g$.

Denote by $H^{-1}(D)$ the dual space of $H^1_0(D)$ and by $ \cjg \cdot,\cdot \cjd_g$ the duality bracket  obtained by  extending the mapping $(f,f')\in C^\infty_0(D)^2\mapsto  \cjg f,f' \cjd_g:=\int_Dff'\,{\rm dv}_g$ to $H^{-1}(D)\times H^1_0(D)$.

The Dirichlet Gaussian Free Field (Dirichlet GFF) $X$ on  $D$ is a Gaussian random variable taking values in   $H^{-1}(D)$  characterized by its mean and covariance kernel for test functions $f,f'\in H^{1}_0(D)$ 
$$\E[ \cjg X,f \cjd_g]=0\quad\text{ and }\quad\E[ \cjg X,f \cjd_g\cjg X,f' \cjd_g]=2\pi  \iint_{M^2}f(x)G(x,y)f'(y){\rm v}_g(\dd x){\rm v}_g(\dd y).$$
With these definitions the Dirichlet GFF does not depend on the metric $g$ in the sense that
$$\big(\cjg X,f \cjd_g\big)_{f\in H^{1}_0(D)}\stackrel{law}{=}\big(\cjg X,fg \cjd_0\big)_{f\in H^{1}_0(D)},$$
where the index $0$ stands for quantities evaluated in the Euclidean background metric $g_0:=|dz|^2$, i.e. with $\omega=0$, and where (with a slight abuse of notations) we have identified in the above expression the metric $g$ with the function $e^{\omega}$. With another harmless abuse of notation, we use   $\int X(z) f(z) \dd^2z$ for $\cjg X,f \cjd_0$. 

Notice that this definition of the Dirichlet GFF extends to simply connected  domains of Riemannian manifolds using local charts and conformal invariance of the Dirichlet Green function.

\subsection*{Domain Markov property}
 Let $M$ be a compact Riemannian manifold without boundary equipped with a metric $g$. Let $D$ be some strict domain of $M$ with a smooth Jordan curve $\mathcal{C}$ as boundary and let $\mu$ be a probability measure of the form $\mu(\dd z):=e^{\theta(z)}\mathcal{H}(\dd z) $ for some continuous function $\theta$ and $\mathcal{H}(\dd z)$ is the one-dimensional Hausdorff measure of the volume form ${\rm v}_g$ restricted to $\mathcal{C}$. Let $X_g$ be the GFF on $M$ in the metric $g$. Then the field $\tilde X_g:=X_g-\mu(X_g)$ is a centered Gaussian distribution with covariance kernel given by the Green function $G_\mu$ on $M$ of the Laplacian $\Delta_g$ with  zero average over $\mathcal{C}$ in the $\mu$ measure. In particular it satisfies
\begin{equation}\label{greenmu}
-\Delta_g \int_M G_\mu(\cdot,y)f(y){\rm v}_g(\dd y)=f(\cdot)-\int_{\mathcal{C}}f(y) \mu(\dd y)
\end{equation}
 together with $\int_{\mathcal{C}}G_\mu(\cdot,y) \mu(\dd y)=0$.
 
Consider the harmonic extension operator $P$ defined by
$$\Delta_g Pf=0\text{ in D },\quad Pf_{|\mathcal{C}}=f$$
and denote by $p$ its integral kernel, i.e. $P(f)(x)=\int_{\mathcal{C}}p(x,z)f(z)\mu(\dd z)$. Denote my $P( \tilde X_g)$ the harmonic extension inside $D$ of the boundary values of $\tilde X_g$ restricted to $\mathcal{C}$. We claim 

\begin{proposition}\label{markovP}
The field $\tilde X_g-P( \tilde X_g)$ is a Dirichlet GFF inside $D$ independent of $P( \tilde X_g)$.
\end{proposition}

\begin{proof} Indeed this is a simple consequence of the fact that the covariance kernel of the field $\tilde X_g-P( \tilde X_g)$ is given by
$$\tilde G(x,y):=G_\mu(x,y)-\int_{\mathcal{C}}p(x,z)G_\mu(z,y)\mu(\dd z)-\int_{\mathcal{C}}p(y,z)G_\mu(z,x)\mu(\dd z)+\iint_{\mathcal{C}^2}p(x,z)p(y,z')G_\mu(z,z')\mu(\dd z),$$
which is symmetric and satisfies 
\begin{equation}\label{greenmudir}
-\Delta_g \tilde G(\cdot,y)=\delta_y(\cdot)\quad\text{in }D,\quad \tilde G(\cdot,y)_{|\mathcal{C}}=0.
\end{equation}
Hence $\tilde G$ is the Dirichlet Green function inside $D$. Furthermore $G_\mu-\tilde G$ is harmonic in both variables inside $D$ and  coincides with $G_\mu$ when one of the variables $x,y$ is sent to the boundary $\mathcal{C}$. Hence $G_\mu-\tilde G$ is equal to $\iint_{\mathcal{C}^2}p(x,z)p(y,z')G_\mu(z,z')\mu(\dd z)$, which is nothing but the covariance kernel of $P( \tilde X_g)$.
\end{proof}

The important point to get the above decomposition of the field $\tilde X_g$ is that the spatial average requirement is localized on the boundary $\mathcal{C} $, hence the boundary average in the right-hand side of \eqref{greenmu} does not show up in the right-hand side of \eqref{greenmudir}, hence allowing us to identify the Dirichlet Green function.

\section{Hitting time for Brownian motion}\label{hitting}
Consider a $2d$ Brownian motion $B$ starting from $(0,y)$ and denote $\cT_A$   the hitting time of a set $A$ by $B$. Set $\gD:=\{ (x,x) : x\in \bbR\}$. We claim that for some $C>0$ and all $|y|\leq 1/2$  
\begin{equation}\label{hit1}
 \bP_{(0,y)} \left[ \cT_{\gD}> \cT_{\{-1,1\}\times \bbR}   \right]\le C|x-y|.
\end{equation}
We can assume that $y>0$ without loss of generality. 
Consider the square of side-length $1/2$ with one side being a segment of $\gD$ centered at the orthogonal projection of $(y,0)$ on $\Delta$:
$(y/\sqrt{2}$,$y/\sqrt{2})$ and which contains $(0,y)$. For $y$ sufficiently small
this squares does not touch $\{-1,1\}\times \bbR$ and thus the probability is then bounded above by the probability that $B$ exits the square by not using the side included in $\gD$. Now using the fact that  the projections of $B$ along the diagonal and in the orthogonal direction are independent this latter  probability  can be computed using standard estimates for the one dimensional Brownian Motion.

\section{Proof of Proposition \ref{lasuperbelleproposition}} \label{siu}
 Let us reformulate the result. The concentration estimate can be deduced from a control of the Laplace transform of $\tilde M^t(\cA)$.
 We first prove a result for  $M^{0,(t)}$ which is obtained by replacing $\tilde X- \tilde X_t$ by $X^{0}-X^{0}_t$,
where $X^{0}$ is the field associated with the special covariance function   (recall \eqref{inzeform}) 
 $Q^{0}_u(x,y)=\rho(e^{u}(x-y))$ for some fixed positive continuous  function $\rho$ with $\rho(0)=1$ and with support included in the ball of radius $1$.
Or more precisely we prove that for  $p\in (1, 4/\gamma^2)$ all  $r \le e^{2t}$
\begin{equation}\label{theright2}
  \bbE\left[ \exp\left( -r   \left(M^{0,(t)}(\cA)-|\cA|\right)\right)\right]\le \exp\left( c_p r^p  e^{2t(1-p)}|\cA| \right).
\end{equation}
From this result we can obtain some information on the Laplace transform of $\tilde M^t(\cA)$ using Kahane's convexity inequality.
Setting $t':=  t- \alpha \log t+c$ (with $c$ a constant and $\alpha>1/\upsilon$, $\upsilon$ being the exponent appearing in \eqref{intfunction}) and $\gep(t)= e^{-2t}$, 
we have 
\begin{equation}
  \int^{\infty}_t  Q_u(x,y)\,\dd u\le \int^{\infty}_{t'} Q^{0}_u(x,y)\,\dd u+\gep(t),
\end{equation}
(the case where $|x-y|\le  e^{-t}  t^{\alpha}$ can be handled using \eqref{asymptex} which has to be valid on both sides and 
the other case is handled by \eqref{intfunction}). Hence if $Z$ is a standard Gaussian independent of $X^{0}$ 
we have 
\begin{multline}
  \bbE\left[ \exp\left( -r  \tilde M^{(t)}(\cA)\right)\right] \le 
  \bbE\left[ \exp\left( -r e^{\gamma\gep(t)^{1/2}Z-\frac{\gamma^2\gep(t)}{2}} M^{0,(t')}(\cA)\right)\right]\\
  \le   \bbE\left[ \exp\left( - \frac{9r}{10} M^{0,(t')}(\cA) \right)\right]+\bbP\left[ Z \ge c\gep(t)^{-1/2} \right].
\end{multline}
The first term can be estimated using \eqref{theright2} and the second one is of order $\exp(- c e^{-2t})$.
Thus we obtain for $r \le e^{2t'}$
\begin{equation}
   \bbE\left[ \exp\left( -r  \tilde (M^{(t)}(\cA)-|\cA|)\right)\right]\le   \exp\big(\frac{r}{10}|\cA|+ c_p r^p  e^{2t(1-p)}|\cA| \big)  +\exp(r|\cA|- c e^{-2t}).
 \end{equation}
 We deduce (using $|\cA|\leq 1$)
\begin{align*}
\P\big( \tilde M^{(t_0)}(\cA)\leq |\cA|/2\big)\leq & e^{-\frac{r}{2}|\cA|}   \bbE\left[ \exp\left( -r    (\tilde M^{(t_0)}(\cA)-|\cA|)\right)\right]\\
\leq &  \exp\big(-\frac{2r}{5}|\cA|+ c_p r^p  e^{2t(1-p)}|\cA| \big)  +\exp(r|\cA|/2- c e^{-2t_0}|\cA|),
\end{align*}
which applied to $r= \delta e^{2t'_0}$ for $\delta$ sufficiently small yields the result.

\medskip

Now to prove \eqref{theright2} we can split $\cA$ according to the intersection with squares  of side-length $e^{-t}$ forming a partition of $D$.   Let   $(D^{t}_i)_{i\in \lint 1,4q \rint}$ with  $q_t=e^{2t}/4$ be such squares (say we fix the value of $t$ so that $e^{t}$ is an even integer
).  Now we have 
\begin{equation}
 M^{0,(t)}(\cA)-|\cA|= \sum_{i=1}^{4q} M^{0,(t)}(\cA \cap D^{t}_i)-|\cA \cap D^{t}_i|=: \sum_{i=1}^{4q} Z_i.
\end{equation}
Now using independence of the field at distance $t$ we can split the sum above in four groups of independent random variables. Reordering the indices we can assume that these four groups 
are $(1,\dots,q)$, $(q+1,\dots, 2q)$...
We have from Jensen's inequality
\begin{equation}
 \bbE\left[ \exp\left( -r \left( M^{0,(t)}(\cA)-|\cA|\right) \right)\right]\le 
 \left( \prod_{j=1}^4 \bbE\left[ \exp\left( -4 r  \sum_{i=1}^{q} Z_{qj+i}    \right)\right]      \right)^{1/4}.
\end{equation}
To conclude it is thus sufficient to control the exponential moment of $\sum_{i=1}^{q} Z_{qj+i}$ which is a sum of independent random variables.
The important input is to be able to control the $p$ moment of $Z_i$.

 \begin{lemma}\label{lpmom}
 Denote by $M$ the GMC associated to the field $X^0$. If $\cB$ is a subset of $D$ then we have for every $1<p<4/\gamma^2$
 \begin{equation}
  \bbE \left[ \left( M (\cB)-|\cB|\right)^p \right] \le C_p |\cB|.
 \end{equation}

\end{lemma}

\begin{proof}
It is sufficient to prove the bound for $\bbE \left[  |M (\cB)|^p \right]$
and then use that $|a+b|^p\le 2^{p-1}(|a|^p+|b|^p)$.

Now we have by Cameron-Martin

\begin{equation}
 \bbE \left[  M (\cB)^p\right]\le \bbE\left[ M(\cB) M(D)^{p-1}\right]=\int_{\cB} \bbE_x\left[M(D)^{p-1}\right] \dd x,  
\end{equation}
where $\bbP_x$ denote tilted measure where $\gamma Q(x, \cdot)$ is added to $X$. One concludes by showing that  $\bbE_x\left[M(D)^{p-1}\right]$ is uniformly bounded in $x$, which is done in \cite[Section 3.4]{DKRV}.
\end{proof}

Setting $\cA_i:=\cA \cap D^{(t)}_i$, scaling (the field $X^0-X^0_t$ has same law as $X^0(e^t\cdot)$) yields  immediately  
\begin{equation}\label{down}
\bbE \left[ |Z_i|^p \right]= e^{-2t(p-1)}\E[M(\cA_i)^p]\le C_pe^{-2t(p-1)}|\cA_i|.
\end{equation}
 
Now we use the relation
$$ e^{x}-x-1\le c_p |x|^p $$ valid for all $x\in (-\infty,4]$. Then for  $r\le e^{2t}$ (using the fact that $Z_i\ge -|\cA_i|\ge -e^{-2t}$),
we have 
\begin{equation}
\bbE\left[ \exp\left( -4r  \sum_{i=1}^{q} Z_{qj+i}    \right) \right]
\le \exp\left(c_p r^p\sum_{i=1}^q\bbE\left[|Z_i|^p\right]\right),
\end{equation}
Using \eqref{down} this yields \eqref{theright2}
 \begin{equation}
  \bbE\left[ \exp\left( -r  \big(M^{0,(t)}(\cA)-|A|\big) \right)\right]
  \le \exp\left( c_p r^p e^{-2t(p-1)}|\cA| \right).
 \end{equation}
 
 \section{Motivations: $2d$ quantum gravity off conformal invariance?}\label{maps}
We stress from the very beginning that this section is entirely speculative or conjectural from a mathematical perspective. Our motivations for constructing a path integral for Mabuchi K-energy find root in a deeper understanding of $2d$ Euclidean quantum gravity. We adopt a point of view  advocated by A. Bilal, F. Ferrari, S. Klevtsov and S. Zelditch in a series of works \cite{BFK,FKZ1,FKZ2}, which we recast  in terms of scaling limit of Random Planar Maps (RPM)  in order to be more adapted to a probabilistic readership.

Recall that RPM have been introduced as a way of discretizing $2d$-quantum gravity. For that one usually considers    planar lattices that can be embedded onto a compact Riemann surface $M$ (without boundary), which we require to be of  genus $\mathbf{h}\geq 2$.  To fix the ideas, we consider finite triangulations of $M$ as our lattices. So  let $\mathcal{T}_{N}$ be the (finite) set of triangulations of $M$   with $N$ faces (up to orientation preserving homeomorphisms). Now we need to embed conformally such lattices onto $M$. There is a subtlety here: there are infinitely many (non diffeomorphic) conformal structures on $M$ as its genus is higher than $2$. We wish to get rid of this degree of freedom. Let us recall how it goes. 

Let ${\rm Met}(M)$ be the space of Riemannian metric on $M$. Two metrics $g,g'\in {\rm Met}(M)$ are said equivalent if 
$$g'=\psi^*(e^{\omega}g)$$
where $\omega\in C^\infty(M)$ and $\psi$ a diffeomorphism ($\psi^*$ is the pushforward). Let $ \mathcal{M}$ be the set of equivalence classes, called the moduli space. Uniformization theorem tells us that each equivalence class contains a metric $g$ with uniformized scalar curvature $K_g=-2$ (called hyperbolic metric). So we consider a fixed family $(g_\tau)_{\tau\in \mathcal{M}}$ of hyperbolic metrics on $M$ parameterized by moduli $\tau\in \mathcal{M}$. 

The procedure to embed triangulations conformally onto $M$ is now the following.  Each  $T\in \mathcal{T}_{N}$ can be equipped with a metric structure $h_T$, where each triangle is given volume $1/N$. The metric structure consists in gluing flat equilateral triangles: the exact definition of the metric structure is given in the lecture notes \cite{Houches} in the case of the sphere and the case we consider here does not present additional difficulties.   Then the uniformization theorem tells us that for each $T$  there exists a unique $\tau_T \in \mathcal{M} $ ($\tau_T$ is called the modulus of $T$) along with an orientation preserving diffeomorphism $\psi_T: T \to M$ and a conformal factor $\omega_T\in C^\infty(M)$ such that \footnote{Recall that in the decomposition \eqref{decompmetric}, the functions $\omega_T$ and $\psi_T$ are unique except if the metric $g_{\tau_T}$ possesses non trivial isometries. This situation is rather unlikely to happen as the the set of such moduli has  measure $0$ with respect to a natural measure on $\mathcal{M}$ called  the Weil-Petersson volume form. We could pursue our discussion while including these special moduli but, for simplicity, we  exclude this situation by restricting to those $\tau$ such that $g_\tau$ has trivial isometry group.
}
\begin{equation}\label{decompmetric}
h_T= \psi_T^*(e^{\omega_T} g_{\tau_T}  ).
\end{equation}
 In what follows, we wish to work with triangulations with fixed modulus so we introduce $\mathcal{T}_{N,\tau}$   the set of triangulations of $M$ with $N$ faces and modulus $\tau$, namely those $T$ such that $\tau_t=\tau$.

Now we explain how to couple quantum gravity to matter fields which stand for models of statistical physics, whose partition function denoted $\mathcal{Z}(T)$  can be defined on each triangulation $T\in \mathcal{T}_N$, hence every $\mathcal{T}_{N,\tau}$. Call $\mathcal{Z}_{N,\tau}$ the partition function of the matter field on  triangulations of size $N$ 
\begin{equation}
\label{defZN}
\mathcal{Z}_{N,\tau}=\sum\limits_{T \in \mathcal{T}_{N,\tau}} \mathcal{Z}(T).
\end{equation}
The main point is to determine the scaling limit of the random geometry on $M$ induced by this model as $N\to\infty$. Random geometry is understood in the following sense (further details can be found in \cite{DKRV}): given $N$,  we can pick a  $T\in \mathcal{T}_{N,\tau}$ at random by defining a probability law
\begin{equation}\label{probaZ}
\P_{N,\tau}(T):=\frac{\mathcal{Z}(T)}{\mathcal{Z}_{N,\tau}}. 
\end{equation}
 It induces via \eqref{decompmetric} a random function $\omega_T$ on $M$ and the question is to find what is the limiting law for such a random function sampled according to $\P_{N,\tau}(T)$ as $N\to \infty$. The answer is in most cases unknown (even heuristically) as it depends crucially on the choice of the matter field.

\medskip
 Yet, starting from the eighties with the seminal work \cite{Pol}, physicists have designed a  toolbox to guess what the limiting law should be. It is based on the way the matter field reacts to local changes of geometries. We base the following discussion on the situation when the partition function $\mathcal{Z}(T):=\mathcal{Z}(e^{\omega_T} g_{\tau_T})$ for matter fields is expressed in terms of regularized determinant of Laplacian and Gaussian integrals. For $g\in {\rm Met}(M)$, consider the formal Gaussian integral
\begin{equation}\label{GI}
\mathcal{Z}_{q,m}(g):=\int_{\Sigma} e^{-\frac{1}{4\pi}\int_M\big(|dX|_{ g}^2+i qK_{ g}X+m^2X^2 \big)\,{\rm dv}_{  g}}DX
\end{equation}
where $q\in\R$ is a parameter called background charge $m$ is a mass parameter, $\Sigma$ is a functional space of maps $X:M\to \R$ and $DX$ stands for the formal Lebesgue measure on $\Sigma$.
In the case $q=m=0$, this integral is ill-defined because of the zero-mode divergence (the null eigenvalue of Laplacian) but can be given sense by $\zeta$-regularization \cite{cf:RS}: it results that
\begin{equation}\label{partGFF}
\mathcal{Z}_{q=0,m=0}(g):=\Big(\frac{{\det} '(-\Delta_g)}{V_g}\Big)^{-\mathbf{c}_{\rm mat}/2}
\end{equation}
where ${\det} '(-\Delta_g)$ has been defined in subsection \ref{regdet} and $\mathbf{c}_{\rm mat}=1$. This is the usual interpretation of the partition function of Gaussian Free Field. More generally, when $\mathbf{c}_{\rm mat}\leq 1$, the right-hand side of \eqref{partGFF} can be considered as a relevant partition function for matter fields (for $\mathbf{c}_{\rm mat}=-0$, it corresponds to pure gravity and for $\mathbf{c}_{\rm mat}=-2$ to uniform spanning trees ...). Such an expression has an interesting metric dependence called {\it gravitational anomaly}: if $\hat g$ is another metric conformal to $  g$ then \eqref{detpolyakov} entails
\begin{equation}\label{weylCFT1}
\ln \frac{\mathcal{Z}_{\rm mat}(\hat g) }{\mathcal{Z}_{\rm mat}(g)}=  \frac{\mathbf{c}_{\rm mat}}{96\pi}S^{{\rm cl},0}_{{\rm L}}(\hat g,g) .
 \end{equation}
For symmetry reasons \footnote{Called background independence.},  the limiting law of the random function $\omega_T$ must "balance"   \footnote{Here we are voluntarily vague to keep the discussion reasonably short. Further  explanations can be found for instance in \cite{BFK}.} this gravitational anomaly. Liouville CFT \eqref{Liouvmeasintro} is then expected to describe the limiting law of $\omega_T$ as it is the only "reasonable QFT" able to counterbalance Liouville type gravitational anomalies, the parameter $\gamma$ in \eqref{QLiouville:intro} being tuned in terms of $\mathbf{c}_{\rm mat}$ through the famous central charge balance\footnote{Recall that $1+6Q^2$ is the central charge of Liouville CFT. The ghost contribution is $-26$, see \cite{GRV} for a brief mathematical description of the ghost system and other references.}
$$1+6Q^2-26+\mathbf{c}_{\rm mat}=0.$$
This ansatz has been successfully applied to the coupling of quantum gravity with matter fields satisfying \eqref{weylCFT}, hence CFTs (e.g. Ising model, lattice Gaussian Free-Field, the O($N$) dilute and dense loop models with $0\le N<2$, etc... we refer to \cite{Kos} for a review and references). 

\medskip
Very little is known about models of 2d quantum gravity beyond this CFT framework. Yet, the study of gravitational anomalies provides serious hints about the  type of  path integrals  ruling the limiting law of the random function $\omega_T$.
The  simplest  non conformal QFT is probably the massive Gaussian Free  Field  \eqref{GI}. Computing the gravitational anomaly for this model is unclear at fixed mass $m>0$ but it can be understood perturbatively in the limit $m\to 0$ \cite{FKZ1}. Again, there is a zero-mode divergence which can be removed by restricting the path integral to the functional space $\Sigma:=\{X\, :\, \int_M X\,{\rm dv}_g=0\}$. Let us call
$$\mathcal{Z}_{q,m=0}(g):=\lim_{m\to 0}\mathcal{Z}_{q,m}(g)$$ the resulting limit. For this model  the gravitational anomaly exhibits a Mabuchi K-energy term\footnote{Our example is a simplification of the results in \cite{FKZ1} where they also explore the small mass expansion in the case $q=0$. }: if $\hat g$ is another metric conformal to $  g$ then  
\begin{equation}\label{weylCFT}
\ln \frac{\mathcal{Z}_{q,m=0}(\hat g) }{\mathcal{Z}_{q,m=0}(g)}=  \frac{\mathbf{c}_{\rm mat}}{96\pi}S^{{\rm cl},0}_{{\rm L}}(\hat g,g)   +\frac{q^2(1-\textbf{h})}{4\pi}  S^{\rm cl}_{\rm M}(\hat g,g).
 \end{equation}
This suggests the following conjecture: sample a triangulation of size $N$ according to the probability law \eqref{probaZ} with $\mathcal{Z}(T):= \mathcal{Z}_{q,m=0}(e^{\omega_T} g_{\tau_T})$. 
In the scaling limit as $N\to\infty$ the law of the random function $\omega_T$ is described by the path integral \eqref{MLintro} in the background metric $g_\tau$, conditioned on having volume $1$ (see subsection \ref{sub:string}), with parameters $$\gamma=(4+q^2)^{1/2} -q\quad \text{ and }\quad \beta= \frac{q^2 (\mathbf{h}-1)}{4\pi}.$$

In conclusion,  \eqref{MLintro}  thus appears as  a candidate for modeling fluctuating metrics in quantum gravity coupled to matter field slightly non conformally invariant.  Let us also mention that quantizing Mabuchi K-energy also appears as a collective field theory for Dyson gas \cite{Wiegman1} or Laughlin states in Quantum Hall Effect \cite{Wiegman2} but a precise relation with our path integral is still unclear for us.

\end{document}